\documentclass[12pt]{article}
\usepackage{color}
\usepackage{verbatim}
\usepackage[polutonikogreek,english]{babel}
\usepackage{booktabs,caption}
\usepackage{multicol,multirow,tabularx}
\usepackage{siunitx}
\usepackage{relsize}
\usepackage{multirow}
\usepackage{hyperref}
\usepackage{float}
\usepackage{lineno}
\usepackage[left=2cm, right=5cm, top=2cm]{geometry}
\usepackage{eurosym}
\usepackage{graphics}
\usepackage{epsfig}
\usepackage{amssymb}
\usepackage{amsmath}
\usepackage{adjustbox}
\usepackage{xcolor}
\usepackage{amsthm}
\usepackage{enumitem}

\newtheoremstyle{break}
  {\topsep}{\topsep}%
  {\itshape}{}%
  {\bfseries}{}%
  {\newline}{}%
\theoremstyle{break}
\newtheorem{theorem}{Theorem}
\newtheorem{corollary}{Corollary}
\newtheorem{remark}{Remark}

\usepackage[square,sort,comma,numbers,compress]{natbib}

\DeclareCaptionFont{ninept}{\fontsize{9pt}{11pt}\selectfont #1}
\usepackage[left=2cm, right=5cm, top=2cm]{geometry}
\begin{document}

\title{Coherent structures and bifurcation analysis \\ in a toxin-driven plant-herbivore model}
\author{Grif\'o Gabriele$^{1,2}$, Valenti Giovanna$^{3}$\\
{$^{1}$\small INDAM - Istituto Nazionale di Alta Matematica ``F. Severi", Piazzale Aldo Moro 5, I-00185 Roma, Italy}\\
{$^{2}$\small Department of Mathematics and Computer Science, University of Palermo,}\\
{\small Via Archirafi 34, I-90123 Palermo, Italy.}\\
{$^{3}$\small Department of Engineering, University of Messina, C.da di Dio, I-98166 Messina, Italy.}\\
{\small $^{(\ast)}$Corresponding author: grifo@altamatematica.it}
}
\date{}
\maketitle

\begin{abstract}

This work investigates how toxin-mediated interactions and directed movements shape the emergence of coherent structures in plant-herbivore systems. The analysis focuses on a two-compartment model enclosing a toxin-dependent functional response and a cross-diffusion term that represents ecologically plausible herbivores' movement towards, or away from, vegetation. Two distinct dynamical regimes arise depending on toxicity strength. Under weak toxicity, the system admits at most one biologically feasible coexistence equilibrium, which may lose stability through a Hopf bifurcation generating small-amplitude temporal oscillations. Under strong toxicity, the nonlinear functional response becomes non-monotonic, allowing for multiple coexistence equilibria and abrupt regime shifts. The influence of cross-diffusion on stability is also examined, identifying the conditions under which Turing instabilities and mixed spatiotemporal patterns occur. Near the corresponding bifurcation thresholds, Stuart-Landau amplitude equations are derived via weakly nonlinear analysis, providing a unified framework for the modulation of oscillatory, stationary, and combined Turing-Hopf modes. Numerical simulations corroborate the theoretical predictions, illustrating transitions from spatially uniform states to oscillations, spatial patterns, and mixed behaviour. Overall, this manuscript highlights how chemical defences, nonlinear feedbacks, and movement strategies jointly determine the emergence, selection, and robustness of coherent structures in plant-herbivore systems.\\


\end{abstract}

\section{Introduction}
\label{sec:intro}

Coherent structures, such as stationary patterns, oscillatory dynamics, and mixed spatiotemporal regimes, are a hallmark of many natural systems. Their emergence reflects subtle interactions between nonlinear local processes and mechanisms responsible for spatial redistribution, and they are observed across several fields ranging from chemistry and physics to neuroscience and ecology \cite{MurrayII, Hoyle, Cross, Meron2015, Lacitignola2017, Upadhyay, UPADHYAY2019, Grifo2023, Sensi2023, Gargano2024, MANDAL2025, Consolo2023, Farivar2025, Grifo2025III}. In ecological settings, these structures often underpin the organization of landscapes, shaping how vegetation, herbivores, and resources distribute across space and time. Therefore, understanding the principles driving their formation is essential for interpreting the resilience, adaptability, and long-term behaviour of ecological communities.

Vegetation systems provide a compelling example of this phenomenon. Under the influence of environmental stressors, resource limitations, and consumer pressure, plant biomass can exhibit striking spatial patterns, alternating between high- and low-density regions. A rich variety of mathematical models has been developed to explain these structures, ranging from biomass-water formulations to frameworks incorporating toxicity or multi-factor interactions \cite{Klausmeier1999, UPADHYAY2009, Thakur2016, Eigentler2020, Vidiella2021, Iuorio2023, Carter2024, EIGENTLER2024,  Grifo2025, Grifo2025II, Tadej2025sub, vanderVoort2025sub, ABBAS2025, ZHANG2026, DIN2026}. Despite their differences, these models emphasize that nonlinear feedbacks, whether mediated by resource availability, growth constraints, or plant defences, play a central role in determining when homogeneous configurations remain stable or instead give way to complex self-organized patterns.

Furthermore, trophic interactions introduce an additional layer of dynamical richness. Even in simplified settings where the ecosystem is represented by plants and herbivores alone, the interplay between growth, consumption, and defence mechanisms can give rise to multiple dynamical regimes. In this context, the toxin-dependent functional response introduced in \cite{Li2006, LIU2008, FENG2008, Castillo2012, XIANG2021} highlights how plant chemical defences may inhibit herbivory, altering consumption rates and influencing whether coexistence is stable or prone to oscillations. This mechanism is ecologically grounded: many plants produce secondary metabolites that reduce palatability or induce physiological stress in herbivores, thereby shaping grazing dynamics and population persistence.

Moreover, spatial movement further enriches these dynamics. While classical diffusion captures random motion, empirical observations indicate that herbivores often respond directionally to vegetation gradients. This behaviour can be modelled through a cross-diffusion term, which accounts for attraction toward sparse vegetation under resource scarcity or repulsion when high biomass imposes physical or ecological constraints \cite{UPADHYAY2014, Liu2019, LI2019, Xiong2021, Marick2024, SINGHA2025}. The sign and magnitude of these directed movements generate nontrivial spatial feedbacks, providing a mechanism through which homogeneous landscapes may spontaneously develop heterogeneity.

Motivated by these considerations, the goal of the present manuscript is to investigate how the combined effects of the toxin-mediated functional response and cross-diffusion shape plant-herbivore dynamics, with particular emphasis on the emergence of spatial and spatiotemporal patterns induced by cross-diffusion mechanisms. The manuscript is organized as follows. In Section \ref{sec:model}, the toxin-dependent plant-herbivore model proposed in \cite{Li2006, LIU2008, FENG2008, Castillo2012, XIANG2021} is extended by incorporating cross-diffusion in the herbivore population. In Section \ref{sec:coh}, a comprehensive analysis of the homogeneous equilibria, including the bifurcation structure and the conditions leading to diffusion-driven instabilities, is provided. Here, both weak and strong toxicity regimes are taken into account by highlighting the ecological mechanisms underlying Hopf, Turing, and mixed Turing-Hopf bifurcations. Section \ref{sec:close} focuses on the weakly nonlinear analysis, deriving the amplitude equations that describe the modulation and saturation of coherent structures near the bifurcation thresholds in the 1D setting. Finally, the ecological implications and potential directions for future research are addressed in Section \ref{sec:disc}.

\newpage
\section{The model}
\label{sec:model}

To investigate the emergence of spatially coherent structures in plant-herbivore systems under the influence of plant toxicity and movement behavior, a generalization of the toxin-dependent functional response (TDFR) model, originally proposed by Liu et al. \cite{Li2006, LIU2008, FENG2008, Castillo2012, XIANG2021}, is considered. The classical TDFR framework captures the interplay between plant growth, herbivore consumption, and toxin-mediated inhibition, but assumes independent diffusion for each species. Empirical evidence, however, indicates that herbivore movement is not purely random but strongly influenced by vegetation gradients, leading to cross-diffusion effects. These mechanisms are ecologically relevant because herbivores tend to aggregate in regions of high or low plant biomass depending on environmental stress, altering local grazing pressure and promoting spatial pattern formation \cite{Liu2019, Xiong2021, Marick2024, SINGHA2025}.

According to these assumptions, the model describing the evolution of plant biomass $B(\mathbf{x},t)$ and herbivore density $H(\mathbf{x},t)$ at the position $\mathbf{x} = (x,y) \in \mathbb{R}^2$ and time $t$ is given by 

\begin{equation}
\label{eq:dimmod}
\begin{cases}
\displaystyle \frac{\partial B}{\partial t} = D_B \Delta B + r B \Big( 1 - \frac{B}{k} \Big) - \frac{e B}{1+ h e B} \Bigg( 1 - \frac{e B}{4G (1 + h e B)} \Bigg) H, \\[6pt]
\displaystyle \frac{\partial H}{\partial t} = D_H \Delta H + D_{HB} \Delta B - D H + A \frac{e B}{1+ h e B} \Bigg( 1 - \frac{e B}{4G (1 + h e B)} \Bigg) H,
\end{cases}
\end{equation}
where $\Delta$ is the Laplacian operator. The terms $D_B \Delta B$ and $D_H \Delta H$ represent plant diffusion (seed dispersal or vegetative spread) and random herbivore movements, respectively. The additional term $D_{HB}\Delta B$ introduces cross-diffusion, modeling herbivore displacement toward regions of higher \mbox{($D_{HB}<0$)} or lower \mbox{($D_{HB}>0$)} plant density. This coupling is crucial for capturing realistic spatial heterogeneity and pattern onset. Moreover, vegetation dynamics combine diffusion with logistic growth $rB(1-B/k)$, where $r$ is the intrinsic growth rate and $k$ the carrying capacity. The impact of plant toxicity on the herbivore-plant interaction is described by the toxin-modified Holling type-II functional response

\begin{equation}
C(B) = \frac{eB}{1+heB}\Bigg(1-\frac{eB}{4G(1+heB)}\Bigg),
\label{eq:response}
\end{equation}
where $e$ is the encounter rate, $h$ the handling time, and $G$ the maximum tolerable toxin intake. In particular, the first fraction captures resource saturation, while the second one introduces inhibition by plant toxins, reducing intake at high biomass. Herbivore growth depends on the same response function, scaled by conversion efficiency $A$, which translates consumed biomass into herbivore growth. This formulation reflects the dual role of plant density: promoting herbivore growth at intermediate levels and inducing stress at high densities due to toxin accumulation. 
Figure~\ref{fig:graphical} summarizes the ecological interactions: plant biomass grows logistically and is reduced by herbivore grazing through a toxin-modified response, while herbivore dynamics depend on intake, mortality and toxicity effects. Positive and negative signs indicate reinforcing or inhibitory effects, emphasizing the dual role of plant density in promoting herbivore growth at intermediate levels and inducing stress at high densities.

\begin{figure}[h!]
\centering
\includegraphics[width=0.5\textwidth]{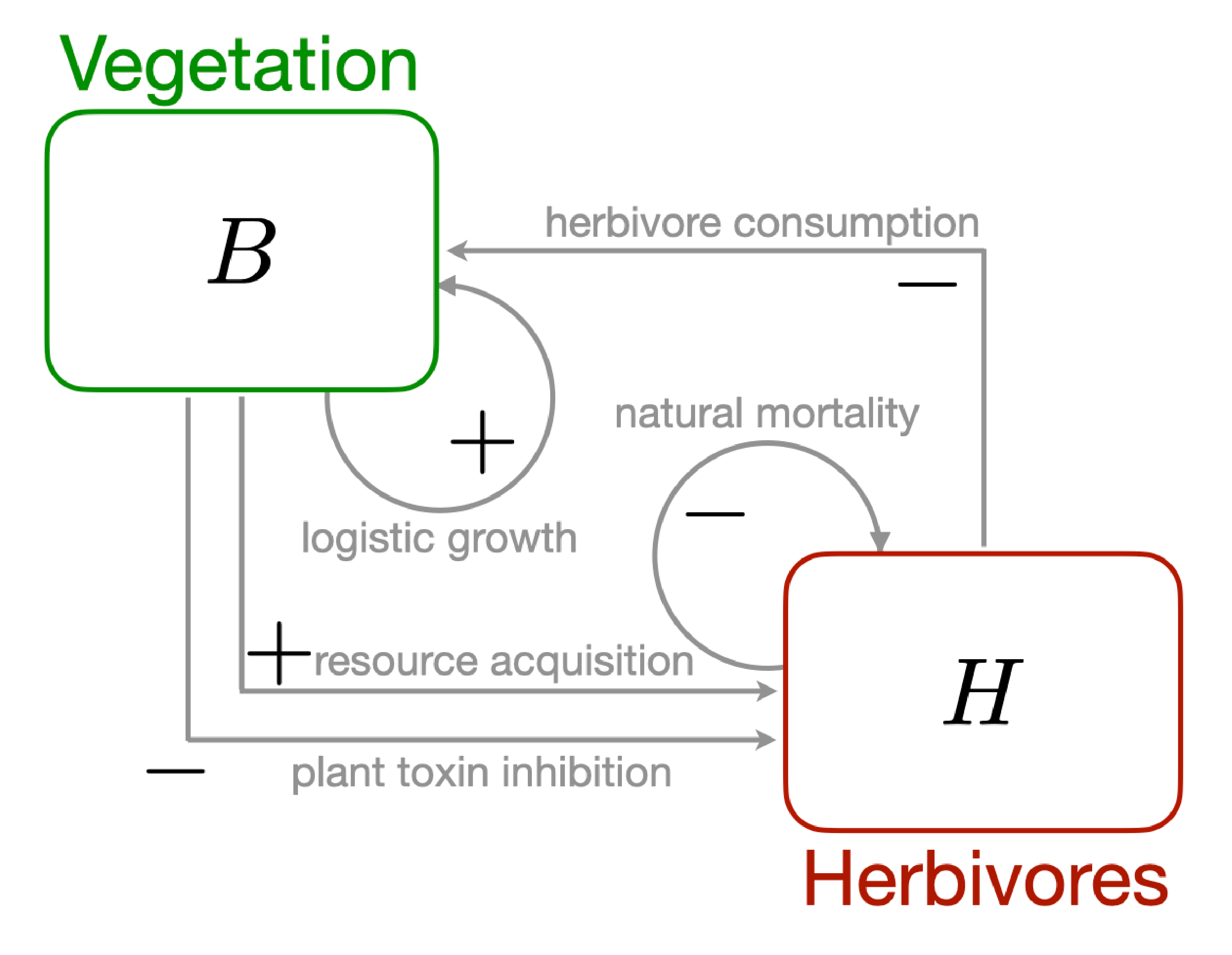}
\caption{Graphical representation of the plant-herbivore system with toxin-mediated interactions. Positive and negative feedbacks are indicated by ``+" and ``-" signs.}
\label{fig:graphical}
\end{figure}

\begin{table}[h!]
\renewcommand{\arraystretch}{1.1}
\centering
\begin{tabular}{c l l}
\hline
Parameter & Description & Units \\
\hline
$r$ & Intrinsic plant growth rate & day$^{-1}$ \\
$k$ & Plant carrying capacity & kg m$^{-2}$ \\
$e$ & Encounter rate & m$^2$ kg$^{-1}$ day$^{-1}$ \\
$h$ & Handling time & day \\
$G$ & Max tolerable toxin intake & day$^{-1}$ \\
$A$ & Conversion efficiency & -- \\
$D$ & Herbivore mortality rate & day$^{-1}$ \\
$D_B$ & Plant diffusion coefficient & m$^2$ day$^{-1}$ \\
$D_H$ & Herbivore diffusion coefficient & m$^2$ day$^{-1}$ \\
$D_{HB}$ & Herbivore cross-diffusion coefficient & m$^2$ day$^{-1}$ \\
\hline
\end{tabular}
\caption{Parameters of the dimensional model \eqref{eq:dimmod}.}
\label{tab:parameters}
\end{table}

To reduce the number of parameters, collected in Table~\ref{tab:parameters}, the following scaling is introduced

\begin{equation}
\begin{array}{c c c c c}
t = \frac{1}{r} \hat{t}, & \mathbf{x} = \sqrt{\frac{r}{D_B}} \hat{\mathbf{x}}, & B = \frac{1}{he} \hat{B}, & H = \frac{A}{he} \hat{H}, & \alpha = ehk, \medskip \\
m = \frac{A}{hr}, & \mu = \frac{1}{4hG}, & d = \frac{D_H}{D_B}, & d_{HB} = \frac{D_{HB}}{A D_B}, & \theta = \frac{D}{r},
\end{array}
\end{equation}
where $\mu \in \mathbb{R}^+$ measures toxicity strength, $d_{HB} \in \mathbb{R}$ the intensity of cross-diffusion, and $d \in \mathbb{R}^+$ the relative herbivore mobility. Moreover, the parameter $\alpha \in \mathbb{R}^+$ denotes the scaled carrying capacity, $m \in \mathbb{R}^+$ the conversion efficiency, and $\theta \in \mathbb{R}^+$ the normalized herbivore mortality rate. Note that, to reflect realistic levels of plant chemical defence \cite{Li2006, LIU2008, FENG2008, Castillo2012, XIANG2021}, the toxicity parameter $\mu$ varies within $\mu \in [1/4,1]$. Therefore, dropping ``\hspace{0.2cm}$\widehat{}$\hspace{0.2cm}" for simplicity, the rescaled model reads as

\begin{equation}
\label{eq:adimmod}
\begin{cases}
\frac{\partial B}{\partial t} = \Delta B + B \left( 1 - \frac{B}{\alpha} \right) - \frac{m B}{1+ B} \left( 1 - \frac{\mu B}{1 + B} \right)H, \medskip \\
\frac{\partial H}{\partial t} = d \Delta H + d_{HB} \Delta B - \theta H + \frac{m B}{1+ B} \left( 1 - \frac{\mu B}{1 + B} \right)H.
\end{cases}
\end{equation}

The aim of this manuscript is to analyze how toxin-mediated responses, cross-diffusion, and nonlinear interactions drive spatial and temporal complexity in plant-herbivore systems. Specifically, the objectives are: (i) to characterize steady states and their stability under varying $\mu$, $\theta$, and $d_{HB}$; (ii) to identify bifurcation structures, including Hopf, Turing, and mixed-mode bifurcations, that govern transitions from homogeneous equilibria to patterned states; (iii) to determine parameter regimes where stationary spatial patterns and spatiotemporal oscillations emerge; and (iv) to explore complex dynamics, such as bistability, which may lead to abrupt shifts between coexistence and herbivore extinction.

\section{Formation of coherent structures}
\label{sec:coh}

This section addresses the emergence of spatially organized patterns in the plant-herbivore system \eqref{eq:adimmod} by analyzing the stability of homogeneous equilibria and identifying the bifurcation mechanisms that lead to their destabilization. Starting from the linearized formulation of the reaction-diffusion model \eqref{eq:adimmod}, the stationary states are characterized and the conditions under which they lose stability through two distinct routes are determined: Hopf bifurcation, associated with temporal oscillations, and Turing instability, responsible for the onset of spatial heterogeneity.

To this aim, the model \eqref{eq:adimmod} is recast in matrix form

\begin{equation}
\mathbf{U}_t = \text{M} \Delta \mathbf{U} + \mathbf{N}(\mathbf{U}),
\label{eq:modmat}
\end{equation}
where the field vector $\mathbf{U}$, diffusion matrix $\text{M}$, and nonlinear kinetic vector $\mathbf{N}(\mathbf{U})$ are given by

\begin{equation}
\mathbf{U} =
\begin{bmatrix}
B \\[4pt]
H
\end{bmatrix}, \quad
\text{M} =
\begin{bmatrix}
1 & 0 \\[4pt]
d_{HB} & d
\end{bmatrix}, \quad
\mathbf{N}(\mathbf{U}) =
\begin{bmatrix}
f(B,H) \\[4pt]
g(B,H)
\end{bmatrix},
\label{eq:modmat_specifico}
\end{equation}
with

\begin{equation}
\begin{aligned}
f(B,H) &:= B\Big(1-\frac{B}{\alpha}\Big) - m \overbrace{\frac{B}{1+B}\Big(1-\frac{\mu B}{1+B}\Big)}^{C(B)}H, \\[4pt]
g(B,H) &:= \Bigg[m \underbrace{\frac{B}{1+B}\Big(1-\frac{\mu B}{1+B}\Big)}_{C(B)} - \theta\Bigg]H.
\end{aligned}
\label{eq:func_adim}
\end{equation}

Let us now look for the spatially homogeneous steady states $\mathbf{U}^\ast = (B^\ast, H^\ast)$ satisfying $\mathbf{N}(\mathbf{U}^\ast) = \mathbf{0}$. Depending on the toxicity parameter $\mu$ relative to the threshold $\mu_{ex} = m/(4\theta)$, the system admits the following equilibria

\begin{equation}
\begin{array}{l c l}
\text{for} \hspace{0.3cm} \mu < \mu_{ex} & \Rightarrow & \mathbf{U}_1 = (0, 0), \hspace{0.2cm} \mathbf{U}_2 = (\alpha, 0), \hspace{0.2cm} \mathbf{U}_{3,4} = (B_{3,4}, H \left( B_{3,4} \right)), \medskip \\
\text{for} \hspace{0.3cm} \mu = \mu_{ex} & \Rightarrow & \mathbf{U}_1 = (0, 0), \hspace{0.2cm} \mathbf{U}_2 = (\alpha, 0), \hspace{0.2cm} \mathbf{U}_{3}=\mathbf{U}_{4}, \medskip \\
\text{for} \hspace{0.3cm} \mu > \mu_{ex} & \Rightarrow & \mathbf{U}_1 = (0, 0), \hspace{0.2cm} \mathbf{U}_2 = (\alpha, 0),
\end{array}
\label{eq:steadystates}
\end{equation}
where 

\begin{equation}
B_{3,4} = \frac{m - 2 \theta \mp \sqrt{m(m - 4 \mu \theta)}}{2(\theta - m + \mu m)} \hspace{0.5cm} \text{and} \hspace{0.5cm} H \left( B_{3,4} \right)= \frac{B_{3,4} \left(\alpha - B_{3,4}\right)}{\alpha m C(B_{3,4})}. 
\label{eq:B34}
\end{equation}
Note that the steady states $\mathbf{U}_1$ and $\mathbf{U}_2$ correspond to extinction and vegetation-only states, while $\mathbf{U}_{3,4}$ represent coexistence configurations whose feasibility depends on toxicity, conversion efficiency, and mortality. These equilibria are biologically admissible only under $0 < B_{3,4} < \alpha$, ensuring vegetation density remains below carrying capacity without exceeding its environmental limit. This constraint reflects ecological feasibility and will be further analyzed in Section~\ref{subsec:temporal}.

To assess stability, a small perturbation around $\mathbf{U}^\ast$ is considered in the form of a Fourier mode

\begin{equation}
\mathbf{U}=\mathbf{U}^\ast+\widetilde{\mathbf{U}}e^{\omega t+i\mathbf{k}\cdot\mathbf{x}},
\label{eq:fourier}
\end{equation}
where $\omega$ is the growth rate and $\mathbf{k}$ the wavenumber. Substituting \eqref{eq:fourier} into \eqref{eq:modmat} yields the linearized problem

\begin{equation}
\big[\omega\text{I}+\text{M}k^2-\mathcal{L}^\ast\big]\widetilde{\mathbf{U}}=\mathbf{0},
\label{eq:linprob}
\end{equation}
being $\text{I}$ the identity matrix, $k=| \mathbf{k}|$, and $\mathcal{L}^\ast=(\nabla_{\mathbf{U}}\mathbf{N})$. Nontrivial solutions exist if the dispersion relation holds

\begin{equation}
\omega^2-\big[g_H^\ast+f_B^\ast-(d+1)k^2\big]\omega+f_B^\ast g_H^\ast-f_H^\ast g_B^\ast-\big(g_H^\ast+df_B^\ast-d_{HB}f_H^\ast\big)k^2+dk^4=0,
\label{eq:disprel}
\end{equation}
where the asterisk denotes the evaluation of the functions at $\mathbf{U}^\ast$ and subscripts indicate partial derivatives.

\subsection{Hopf instability and temporal oscillations}
\label{subsec:temporal}

Once diffusion is neglected ($k = 0$), the stability conditions of the equilibria \eqref{eq:steadystates} are given by

\begin{equation}
\begin{cases}
\text{tr}(\mathcal{L}^\ast) = f_B^\ast + g_H^\ast < 0, \medskip \\
\text{det}(\mathcal{L}^\ast) = f_B^\ast g_H^\ast - g_B^\ast f_H^\ast > 0,
\end{cases}
\label{eq:cond_stab}
\end{equation}
By evaluating the above quantities at the equilibria $\mathbf{U}_1$ and $\mathbf{U}_2$, the conditions \eqref{eq:cond_stab} reduce to

$$
\begin{cases}
\text{tr}(\mathcal{L}^\ast) (\mathbf{U}_1) = 1 - \theta < 0 \medskip \\
\text{det}(\mathcal{L}^\ast) (\mathbf{U}_1) = - \theta >0
\end{cases},
\hspace{1cm}
\begin{cases}
\text{tr}(\mathcal{L}^\ast) (\mathbf{U}_2) = m C(\alpha) - \theta - 1 < 0\medskip \\
\text{det}(\mathcal{L}^\ast) (\mathbf{U}_2) = \theta - m C (\alpha) > 0
\end{cases}
$$
so that, being $\theta > 0$, it follows that $\mathbf{U}_1$ is always unstable, whereas the stability of $\mathbf{U}_2$ depends on the model parameters. More precisely, it is stable under homogeneous perturbations iff

\begin{equation}
C(\alpha) = \frac{\alpha}{(1+\alpha)^2} \left[1 + \alpha \left(1 - \mu\right) \right]  < \frac{\theta}{m}
\hspace{0.5cm}
\Leftrightarrow
\hspace{0.5cm}
\mu > \mu^\ast, \quad \text{with} \quad
\mu^\ast = \frac{(1+\alpha)\big[m\alpha - \theta(1+\alpha)\big]}{m\alpha^2}.
\end{equation}
Moreover, the biological admissibility of the coexistence equilibria $\mathbf{U}_{3,4}$ depends critically on the qualitative behavior of the toxin-modified functional response $C(B)$ introduced in \eqref{eq:func_adim}. On the one hand, for $\mu \in [1/4, 1/2]$, inhibition is weak and $C(B)$ is monotonic, approaching $C(B) = m(1-\mu)$ as $B \to \infty$ (see Figure~\ref{fig:C_behavior}(a)). On the other hand, for $\mu \in [1/2, 1]$, the function becomes unimodal, reaching a maximum at $B_M = 1/(2\mu - 1)$ before decreasing toward $C(B) = m(1-\mu)$ (see Figure~\ref{fig:C_behavior}(b)). This transition reflects an ecological trade-off: higher plant density initially benefits herbivores, but excessive biomass leads to toxin accumulation and reduced consumption efficiency. Therefore, the subsequent analysis distinguishes between weak and strong toxicity regimes, as they lead to markedly different dynamical scenarios: unimodal responses allow multiple coexistence equilibria and complex dynamics, whereas monotonic responses lead to simpler configurations.

\begin{figure}[t!]
\centering
\includegraphics[width=0.9\textwidth]{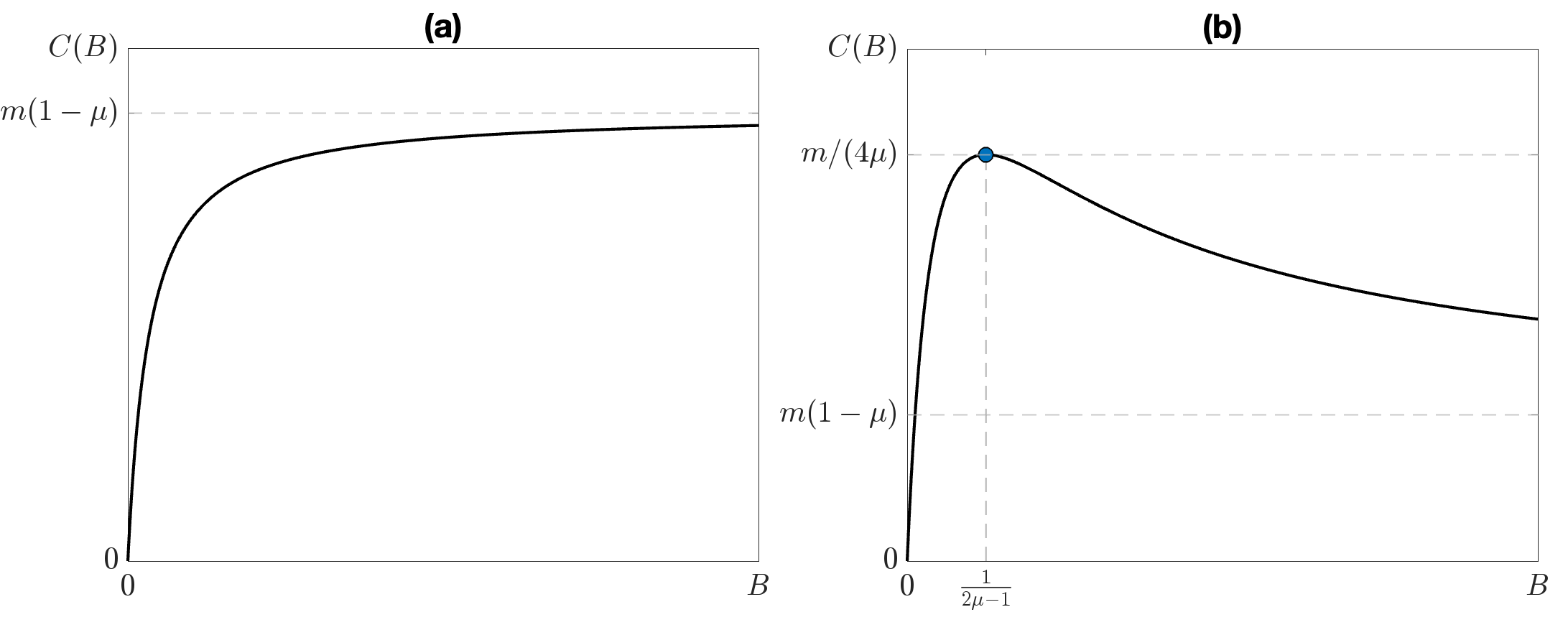}
\caption{Qualitative behavior of the toxin-modified functional response $C(B)$ for two toxicity regimes: (a) low toxicity $\mu \in [1/4,1/2]$, monotonic increasing; (b) high toxicity $\mu \in [1/2,1]$, unimodal with a maximum.}
\label{fig:C_behavior}
\end{figure}

\subsubsection{Weak toxicity regime}
\label{sec:LSA_weak}

When $\mu \in [1/4, 1/2]$, the inhibitory effect of plant toxins is relatively weak, and the functional response $C(B)$ remains monotonic. Specifically, $C'(B) > 0$ for $B > 0$, implying that herbivore intake increases with vegetation density and approaches an asymptote $C(B) = m(1-\mu)$ as $B \to \infty$ (see Fig. \ref{fig:C_behavior}(a)). Ecologically, this scenario reflects environments where chemical defenses are insufficient to significantly reduce consumption at high biomass levels.

Let us now investigate the biological feasibility of coexistence equilibria to exclude non-physical configurations such as negative biomass or vegetation exceeding carrying capacity $\alpha$. Results are collected in the following Theorem.

\begin{theorem}[Existence and stability of the coexistence equilibrium under weak toxicity]
\label{th:weak}
Let us consider system \eqref{eq:adimmod} with $\mu \in [1/4,1/2]$. Then, a unique positive coexistence equilibrium $\mathbf{U}_3=(B_3,H_3)$ with $B_3 \in ]0, \alpha[$ exists iff $\mu < \mu^\ast$. Moreover, it is locally asymptotically stable with respect to homogeneous perturbations iff $\theta>\widetilde{\theta}$, being

\begin{equation}
\widetilde{\theta} = \frac{m \left \{ \left[ 1 + \alpha \left(1-\mu \right) \right]\left[ 1 + \alpha \left(1 + 5 \mu \right) - 4 \mu \right] + \left[ \alpha \left( 3 \mu - 1\right)-1\right] \sqrt{\left[1+\alpha \left(1-\mu \right) \right] \left[ \alpha \left( 1 - \mu \right) + 1 + 8 \mu \right]} \right \} }{8 \mu \left(1+\alpha \right)^2}.
\label{eq:th2_stab}
\end{equation}
\end{theorem}

\begin{proof}
First of all, let us notice that the coexistence equilibria $\mathbf{U}_{3,4}$ are biologically admissible iff  $0<B_{3,4}<\alpha$. On the one hand, it can be observed that $B_3$ is always positive and real for $\mu < \mu_{ex}$, whereas it is bounded by $\alpha$ iff 

$$
\frac{\alpha}{(1+\alpha)^2} \left[1 + \alpha \left(1 - \mu\right) \right] > \frac{\theta}{m},
$$
namely iff $\mu < \mu^\ast$ is satisfied. Therefore, the coexistence equilibrium $\mathbf{U}_3$ exists for \mbox{$\mu < \min\{\mu_{ex},\mu^\ast\}=\mu^\ast$}. On the other hand, it is easy to ascertain that $B_4<0$ for any value of model parameters, leading to the $\mathbf{U}_4$ biological meaninglessness. 

For what concern the stability character of $\mathbf{U}_3$, the conditions \eqref{eq:cond_stab} are reduced to

\begin{equation}
\begin{cases}
\theta H'(B_3) < 0, \medskip \\
m \theta C'(B_3) H(B_3) > 0.
\end{cases}
\label{eq:th2_cases}
\end{equation}
Since in the weak toxicity regime $C'(B)>0\hspace{0.2cm} \forall B$, the inequality \eqref{eq:th2_cases}$_2$ is always satisfied, so that the stability of $\mathbf{U}_3$ depends on the sign of $H'(B_3)$. More precisely, the inequality \eqref{eq:th2_cases}$_1$ yields to

\begin{equation}
2(1-\mu)B_3^2+[3-\alpha(1-\mu)]B_3+1-\alpha(1+\mu)>0
\end{equation}
which is satisfied for $\theta>\widetilde{\theta}$.
\end{proof}

\begin{remark}
From an ecological viewpoint, the constraint $\mu < \mu_{ex}$ reflects that toxicity cannot be excessively strong relative to conversion efficiency and mortality; otherwise, the physiological stress induced by toxins would prevent herbivore persistence, making coexistence impossible. Similarly, $\mu < \mu^\ast$ indicates that weak toxicity allows herbivores to exploit vegetation without severe inhibition, sustaining coexistence. As $\mu$ approaches $\mu^\ast$, toxins dominate, reducing intake and driving the system toward extinction. On the other hand, $\widetilde{\theta}$ represents a critical mortality level separating different regimes. If $\theta < \widetilde{\theta}$, the population grows rapidly, increasing grazing pressure and destabilizing vegetation, which may trigger oscillations or even collapse. For $\theta > \widetilde{\theta}$, mortality is sufficiently high to prevent uncontrolled herbivore growth, limiting grazing pressure and restoring stability. Excessive mortality, however, leads to herbivore extinction and a vegetation-only state. Stability thus emerges within an intermediate range balancing resource availability and physiological constraints.
\end{remark}

\noindent When the herbivore mortality rate $\theta$ approaches a critical threshold, the coexistence equilibrium $\mathbf{U}_3$ may lose stability. After this transition, the system no longer guarantees convergence to the original steady state and can exhibit qualitatively different behaviors. Depending on the underlying bifurcation mechanism, trajectories may converge to another stable equilibrium or evolve toward a family of limit cycles, provided by a Hopf bifurcation. Therefore, let us now investigate such a scenario.

\begin{corollary}[Hopf bifurcation under weak toxicity]
\label{cr:hopf_weak}
Under the assumptions of Theorem~\ref{th:weak}, there exists a critical value $\theta=\theta_{\text{hopf}}$ such that the equilibrium $\mathbf{U}_3$ undergoes a Hopf bifurcation. For $\theta$ near $\theta_{\text{hopf}}$, the system admits a family of small-amplitude periodic solutions.
\end{corollary}

\begin{proof}
According to Theorem~\ref{th:weak}, $\mathbf{U}_3$ changes its stability character as a function of $\theta$. Indeed, when $f_B(\mathbf{U}_3)=0$, the trace \eqref{eq:cond_stab}$_1$ vanishes so that the characteristic equation admits a pair of purely imaginary conjugate roots at $\theta=\widetilde{\theta}$ and the equilibrium loses stability. Moreover, since $\widetilde{\theta}>0$, $H''(B_3)\big|_{\widetilde{\theta}}<0$ being $B_3$ a local maximum for $H$ at $\theta=\widetilde{\theta}$, and
$\left.\partial B_3/\partial\theta\right|_{\widetilde{\theta}}>0$, the transversality condition

\begin{equation}
\left.\frac{d\text{Re}\{\lambda(\theta)\}}{d\theta}\right|_{\widetilde{\theta}} = \left.\big[\theta H''(B_3)\,\partial B_3/\partial\theta\big]\right|_{\widetilde{\theta}} < 0
\end{equation}
is always fulfilled, ensuring that the eigenvalues cross the imaginary axis with non-zero speed.

Therefore, the system undergoes a Hopf bifurcation at $\widetilde{\theta} = \theta_{\text{hopf}}$ via a spatially uniform periodic solution whose angular frequency is given by 

\begin{equation}
\omega_H = \sqrt{m \theta_{hopf} C'(B_3) H_3}.
\end{equation}
\end{proof}

\noindent Summarising, under weak toxicity, the system admits at most a coexistence equilibrium. Its dynamics are governed by stability changes via Hopf bifurcation when $\theta = \theta_{\text{hopf}}$, without the possibility of multiple equilibria, a feature distinguishing this regime from the strong toxicity case discussed next.

\subsubsection{Strong toxicity regime}
\label{sec:LSA_strong}

When $\mu \in (1/2,1]$, plant toxins exert a dominant inhibitory effect, and the functional response $C(B)$ becomes unimodal. Herbivore intake initially increases with vegetation density, reaches a maximum at $B_M=1/(2\mu-1)$ with $C(B_M)=1/(4\mu)$, and then declines as toxicity accumulates (see Fig.~\ref{fig:C_behavior}(b)). Ecologically, this regime reflects environments where chemical defenses strongly limit consumption at high biomass, creating a trade-off between resource abundance and physiological stress. 

Results ensuring the existence and stability of coexistence equilibria are recollected in the following Theorem.

%Under strong toxicity, the system can exhibit two distinct coexistence equilibria in addition to the trivial state $\mathbf{U}_1=(0,0)$ and the vegetation-only state $\mathbf{U}_2=(\alpha,0)$. These equilibria correspond to different ecological scenarios and arise from the unimodal shape of the functional response, which creates separate parameter regions for feasibility. The conditions ensuring their existence are summarized in the next Theorem.

\begin{theorem}[Existence and stability of the coexistence equilibria under strong toxicity]
\label{th:strong}
Let us consider system \eqref{eq:adimmod} with $\mu\in(1/2,1]$. Then, the positive coexistence equilibria $\mathbf{U}_3$ exist for $\mu<\mu^\ast$, whereas the steady state $\mathbf{U}_4$ is admitted for $\max\{\hat{\mu},\mu^\ast\}<\mu<\mu_{ex}$, being $\hat{\mu}=\frac{\alpha+1}{2\alpha}$. Moreover, $\mathbf{U}_3$ is locally asymptotically stable with respect to homogeneous perturbations iff $\theta>\widetilde{\theta}$, whereas the equilibrium $\mathbf{U}_4$ is always unstable.
\end{theorem}

\begin{proof}
In order to look for the admissibility of the coexistence equilibria $\mathbf{U}_{3,4}$, let us investigate the location of $B_M$ with respect to $\alpha$. In particular, for $B_M > \alpha$ $\left(\mu < \hat{\mu} = \frac{\alpha+1}{2 \alpha}\right)$, it is easy to ascertain that $B_4 \not \in ]0, \alpha[$ whereas, once $B_M < \alpha$ ($\mu > \hat{\mu}$), the steady state $\mathbf{U}_4$ is biologically admissible for $\mu^\ast<\mu<\mu_{ex}$. On the contrary, whatever the location of $B_M$, the coexistence equilibrium $\mathbf{U}_3$ is biologically admissible iff $\mu < \mu^\ast$.  

The stability conditions \eqref{eq:cond_stab} reduce to

\begin{equation}
\begin{cases}
\theta H'(B_{3,4}) < 0, \medskip \\
m \theta C'(B_{3,4}) H(B_{3,4}) > 0.
\end{cases}
\label{eq:th4_cases}
\end{equation}
For what concern the steady state $\mathbf{U}_4$, $B_4$ lies beyond $B_M$ so that $C'(B_4)<0$, i.e. the second condition \eqref{eq:th4_cases}$_2$ is violated, and $\mathbf{U}_4$ is always unstable. As far as $\mathbf{U}_3$ is concerned, $C'(B_3)>0$ holds, so leaving the stability to $H'(B_3)<0$ as in Theorem~\ref{th:weak}. Therefore, the stability of $\mathbf{U}_3$ is guaranteed for $\theta>\widetilde{\theta}$.
\end{proof}

\noindent When the herbivore mortality rate $\theta$ approaches a critical threshold, the coexistence equilibrium $\mathbf{U}_3$ may undergo a qualitative change in stability. This transition indicates that the system could depart from its previous behavior and the nature of this change depends on the underlying bifurcation mechanism, as investigated in the following Corollary.

\begin{corollary}[Hopf bifurcation under strong toxicity]
\label{cr:hopf_strong}
Under the assumptions of Theorem~\ref{th:strong}, there exists $\theta=\theta_{\text{hopf}}$ such that $\mathbf{U}_3$ undergoes a Hopf bifurcation. For $\theta$ near $\theta_{\text{hopf}}$, the system admits small-amplitude periodic solutions.
\end{corollary}

\begin{proof}
The proof follows the same procedure as in Corollary~\ref{cr:hopf_weak}, verifying that the trace \eqref{eq:cond_stab}$_1$ vanishes while the determinant \eqref{eq:cond_stab}$_2$ remains positive and the transversality condition holds for $\theta = \theta_{hopf}$.
\end{proof}

\medskip
\noindent In summary, strong toxicity creates two disjoint parameter regions where coexistence is feasible, but only one interior equilibrium $\mathbf{U}_3$ can be stable. Indeed, the occurrence of $\mathbf{U}_4$ in the descending branch of the functional response, where herbivore intake is severely inhibited, leads to an accumulation of toxin that prevents herbivore persistence despite abundant vegetation. Consequently, bistability between two coexistence states does not occur; instead, abrupt transitions arise between a stable coexistence equilibrium and the vegetation-only state as parameters cross critical thresholds. Although Hopf bifurcation remains possible, the presence of multiple equilibria profoundly alters the phase space, highlighting contrasting ecological scenarios between weak and strong toxicity regimes.

\subsection{Turing bifurcation and spatial heterogeneity}

The previous stability analysis identified conditions under which homogeneous equilibria remain stable. However, the formation of spatial patterns requires a different mechanism: a state that is stable to uniform perturbations but loses stability when heterogeneity is introduced. This scenario corresponds to a Turing bifurcation, where diffusion and cross-diffusion interact with local dynamics to generate stationary spatial structures \cite{ALI2026, YANG2026}.

From the results obtained in Section~\ref{subsec:temporal}, it emerges that neither the trivial equilibrium $\mathbf{U}_1$ nor the vegetation-only state $\mathbf{U}_2$ can act as precursors for pattern formation, as the former is always unstable and the latter, although potentially stable, lacks herbivores and is therefore biologically irrelevant. Moreover, according to Theorem \ref{th:strong}, $\mathbf{U}_4$ remains unstable throughout its existence domain. Consequently, under both toxicity regimes, the only biologically meaningful candidate for Turing instability is $\mathbf{U}_3$, which represents a coexistence state for plants and herbivores. This equilibrium can be stable under homogeneous perturbations yet destabilized by spatially heterogeneous ones when diffusion and cross-diffusion effects are considered.

\begin{theorem}[Stability under non-homogeneous perturbations]
\label{th:stab_nonhom}
Let us consider system \eqref{eq:adimmod} with $1/4<\mu<\min\{1,\mu^\ast\}$. Then, the coexistence equilibrium $\mathbf{U}_3$ is locally stable with respect to spatially heterogeneous perturbations iff, for every wavenumber $k>0$, the following conditions hold

\begin{equation}
\begin{cases}
\text{tr}(\mathcal{L}^\ast-\text{M}k^2) = \theta H' (B_3) - (d+1)k^2<0,\\[6pt]
\text{det}(\mathcal{L}^\ast-\text{M}k^2) = m \theta C'(B_3) H(B_3) - \theta (d H' (B_3) + d_{HB} )k^2+dk^4>0.
\end{cases}
\label{eq:stab_nonhom}
\end{equation}
\end{theorem}

\begin{proof}
Linearizing system \eqref{eq:adimmod} around $\mathbf{U}_3$ and considering perturbations of the form \eqref{eq:fourier} yield the characteristic equation \eqref{eq:disprel}. Conditions \eqref{eq:stab_nonhom} ensure that, for all admissible $k$, both eigenvalues have negative real parts, guaranteeing local asymptotic stability under spatially heterogeneous disturbances.
\end{proof}

\noindent When these conditions fail for some non-zero wavenumber $k$, the equilibrium loses stability through a diffusion-driven mechanism, as illustrated in the following corollary.

\begin{corollary}[Turing instability threshold]
\label{cr:turing}
Let us consider system \eqref{eq:adimmod} with $1/4<\mu<\min\{1,\mu^\ast\}$ and $\theta>\widetilde{\theta}$. Then, the equilibrium $\mathbf{U}_3$ undergoes a Turing bifurcation when the cross-diffusion coefficient $d_{HB}$ exceeds

\begin{equation}
d_{HB}^T=\frac{2\sqrt{d m\theta C'(B_3)H(B_3)}-d\theta H'(B_3)}{\theta}
\label{eq:turing_thr}
\end{equation}
with the critical wavenumber

\begin{equation}
k_T=\sqrt[4]{\frac{\theta m C'(B_3)H(B_3)}{d}}.
\label{eq:turing_wave}
\end{equation}
\end{corollary}

\begin{proof}
The assumptions of the Theorems \ref{th:weak} and \ref{th:strong} guarantee that the coexistence state $\mathbf{U}_3$ exists and is stable with respect homogenous perturbation iff $\mu < \mu^\ast$ and $\theta > \widetilde{\theta}$. Therefore, the emergence of coherent structures is obtained once the condition \eqref{eq:stab_nonhom}$_2$ is violated together with \mbox{$\partial \text{det}(\mathcal{L}^\ast-\text{M}k^2) / \partial k = 0$}, namely a root of the characteristic polynomial crosses the null-axis for a non null wavenumber via a maximum. It is easy to ascertain that the above conditions lead to the critical value for the bifurcation parameter $d_{HB}^T$ \eqref{eq:turing_thr} as well as for the critical wavenumber $k_T$ \eqref{eq:turing_wave}.
\end{proof}

\begin{remark}
Biologically, Turing instability corresponds to the emergence of vegetation patches and herbivore aggregations. Cross-diffusion plays a decisive role: when herbivores avoid dense vegetation ($d_{HB}>0$), they redistribute toward sparse areas, reducing local grazing pressure and creating spatial contrasts that destabilize uniform states. Conversely, attraction to vegetation ($d_{HB} \leq 0$) suppresses pattern formation. These dynamics mirror real ecosystems where movement strategies and resource distribution interact to shape landscape heterogeneity.
\end{remark}

From an ecological perspective, $d_{HB}>0$ reflects environments characterized by resource abundance and absence of external stressors. In such scenarios, vegetation is sufficiently dense that herbivores do not need to actively search for food; however, high plant biomass imposes mobility constraints, prompting herbivores to aggregate in areas of lower vegetation.
On the contrary, in resource-limited environments, herbivores are attracted to the few vegetated patches available, corresponding to $d_{HB}<0$. Under this assumption, the present model does not predict any onset of patterned structures, likely due to its simplified representation of ecosystem dynamics. Indeed, resource scarcity is often driven by external factors, such as water limitation, which are not accounted for in the current formulation. To capture these processes, the model should be extended by introducing additional variables describing resource availability and its dynamics. This would lead to a system of three or more coupled equations, enabling the description of more complex trophic interactions. 

%For instance, one could consider an extended network where herbivores interact with both vegetation and water, consuming plant biomass and using water for physiological needs. Vegetation, in turn, depends on water for growth and is grazed by herbivores, while water acts as a limiting resource exploited by both species. Such a formulation would capture bidirectional feedbacks among species and resources, providing a more realistic framework to investigate how resource limitation and multi-level interactions influence the onset and stability of spatial patterns in stressed ecosystems.

From a mathematical perspective, the inclusion of cross-diffusion introduces a nonlinear coupling between species movement and resource distribution, enhancing the potential for pattern formation via Turing-like instabilities. Unlike classical reaction-diffusion systems, where spatial heterogeneity typically requires a strong disparity in diffusion rates, the generalized framework considered here predicts that even moderate diffusion ratios can trigger instability when herbivore displacement away from dense vegetation is enough.

\begin{remark}[Mixed spatiotemporal structures]
The loci of Hopf and Turing bifurcations, identified in the previous analysis, may intersect within the parameter space, giving rise to codimension-two points where the system simultaneously loses stability with respect to both temporal and spatial modes. At these critical configurations, the emerging patterns are neither purely stationary nor purely oscillatory, but exhibit mixed spatiotemporal features: they are periodic in space and oscillatory in time with $\omega_H^2 = d ({k_T})^4$. Such scenarios are particularly relevant in ecological contexts, as they describe regimes where vegetation patches and herbivore aggregations not only organize spatially but also fluctuate periodically, reflecting complex feedbacks between resource distribution, chemical defenses, and movement strategies.
\end{remark}

\subsection{Numerical validation and bifurcation-driven pattern}

To validate the theoretical predictions on the existence and stability of equilibria and to examine the onset of spatio-temporal patterns near bifurcation thresholds, the numerical framework is now introduced. Simulations are performed in a one-dimensional setting under parameter regimes representative of weak and strong toxicity, thus highlighting the dynamics induced by these different ecological conditions. All the parameters considered in the following analysis are taken from previously published studies available in the literature \cite{Li2006, LIU2008, FENG2008, XIANG2021}. The analysis focuses on three aspects: (i) stability of homogeneous equilibria; (ii) oscillatory behaviour triggered by Hopf bifurcations; and (iii) diffusion-driven instabilities leading to spatial pattern formation.

\subsubsection{Weak toxicity regime}
\label{subsec:weak}

As shown in Section~\ref{subsec:temporal}, in the weak toxicity regime, the toxin-modified functional response is monotonic and the model \eqref{eq:adimmod} admits at most one biologically feasible coexistence equilibrium. To isolate local interactions between plant biomass and herbivore density, diffusive effects are initially neglected, emphasizing the temporal oscillatory character of the dynamics. The bifurcation structure and stability of equilibria~\eqref{eq:steadystates} are analyzed as the toxicity parameter $\mu \in [1/4,1/2]$. Figure~\ref{fig:weak_regime_nodiff} summarizes the main outcomes in the $(\theta,\mu)$-plane through six panels, where the first row (a)--(c) explores variations in the plant carrying capacity $\alpha$, and the second row (d)--(f) investigates the effect of conversion efficiency $m$. In each panel, the red solid curve denotes the existence locus of the coexistence equilibrium $\mathbf{U}_3=(B_3,H_3)$, while the green dash-dotted curve marks the Hopf threshold $\theta=\theta_{hopf}$ beyond which $\mathbf{U}_3$ loses stability.
From an ecological perspective, the region to the left of the red curve ensures coexistence of vegetation and herbivores, whereas instability of $\mathbf{U}_3$ to the left of the green curve leads to oscillatory dynamics driven by Hopf bifurcation. It is worth noting that the existence domain of $\mathbf{U}_3$ coincides with the instability of the vegetation-only equilibrium $\mathbf{U}_2=(\alpha,0)$, thereby excluding any bistability between these states.

The role of $\alpha$ in the above dynamics can be summarized as follows:
\begin{enumerate}[label=$\alpha_\arabic*)$]
\item \textbf{Low carrying capacity ($\alpha \leq 2/3$)}: the Hopf threshold lies outside the weak toxicity range, resulting in full stability (panel (a)). Limited resources reduce grazing pressure, favouring equilibrium persistence.
\item \textbf{Intermediate carrying capacity ($2/3<\alpha<4/5$)}: the Hopf threshold enters the weak toxicity regime, creating a zone of oscillatory dynamics (panel (b)). Increased resources amplify trophic interactions, making the system prone to cycles.
\item \textbf{High carrying capacity ($\alpha \geq 4/5$)}: the instability region expands significantly (panel (c)), promoting large-amplitude oscillations or even complex behaviour due to strong resource availability.
\end{enumerate}

\noindent Fixing the $\alpha_3$ scenario, the influence of $m$ is then assessed. Both the existence locus and Hopf threshold lie within the weak toxicity regime, partitioning the parameter space into stable and unstable regions. By inspecting the diagrams shown in Figure \ref{fig:weak_regime_nodiff}(d)--(f), it is evident that increasing $m$ enlarges the coexistence domain markedly, while the Hopf curve shifts moderately. Ecologically, higher conversion efficiency enhances herbivore persistence but intensifies grazing pressure, which combined with low mortality may destabilize the system. Conversely, sufficiently high mortality restores stability.

\begin{figure}[b!]
\centering
\includegraphics[width=1\textwidth]{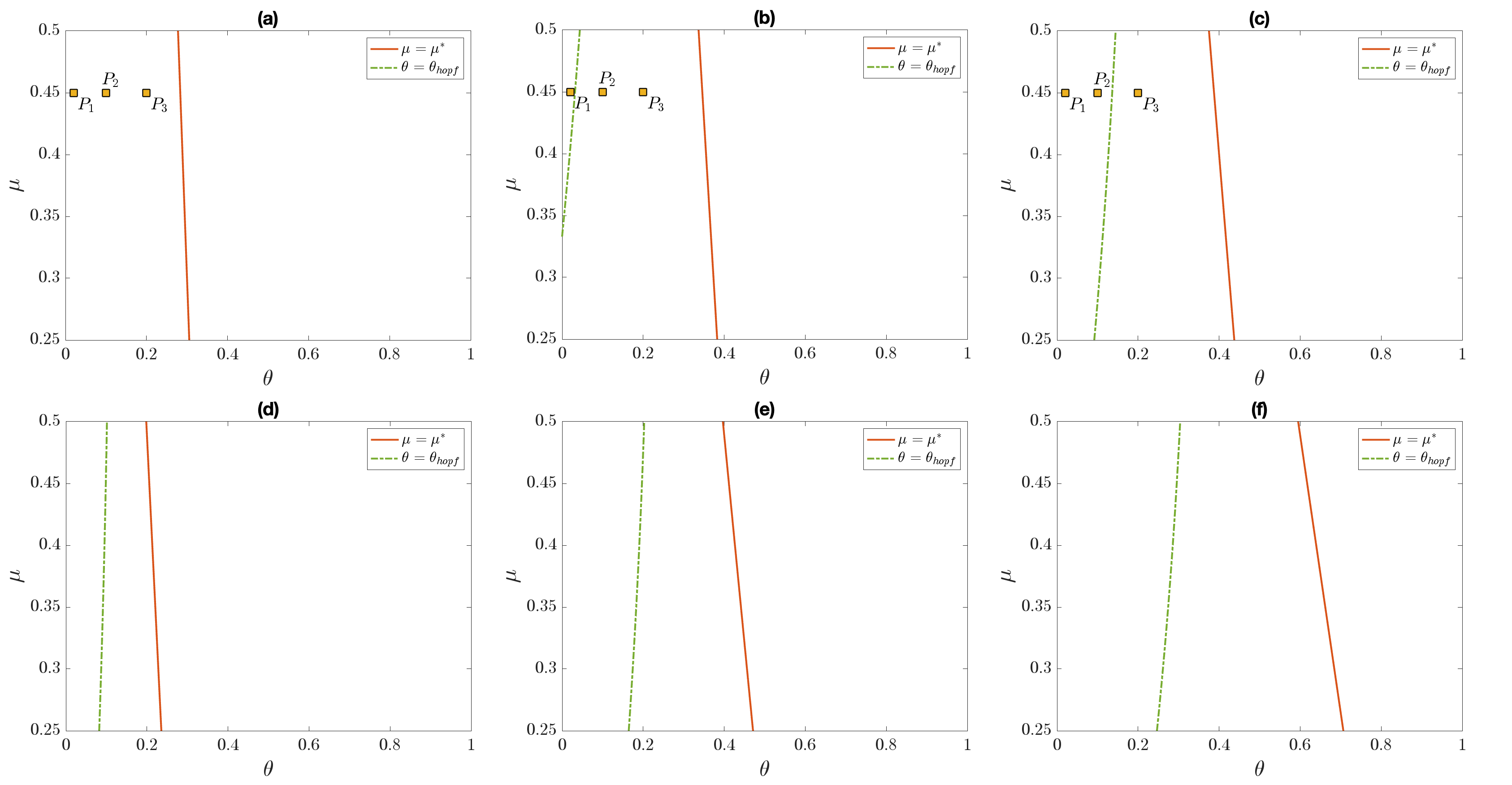}
\caption{Bifurcation diagrams under weak toxicity ($\mu \in [1/4,1/2]$) without diffusion in the $(\theta,\mu)$-plane. Panels (a)--(c): effect of carrying capacity $\alpha$; panels (d)--(f): effect of conversion efficiency $m$. Red solid line: existence locus of $\mathbf{U}_3$; green dash-dotted line: Hopf threshold. Parameter values: $\alpha=0.5$ (a), $\alpha=0.75$ (b), $\alpha=1$ (c), $\alpha=1.2$ (d--f); $m=0.5$ (d), $m=1$ (a--c,e), $m=1.5$ (f). Yellow squares indicate configurations used in subsequent simulations ($P_1$, $P_2$, $P_3$).}
\label{fig:weak_regime_nodiff}
\end{figure}

To validate these theoretical predictions, numerical simulations are performed under the parameter configurations highlighted in Figure \ref{fig:weak_regime_nodiff})(a)--(c) by yellow squares. Figure~\ref{fig:weak_regime_sim} reports the time series of plant biomass and herbivore density following small perturbations of the coexistence equilibrium $\mathbf{U}_3$. Each row corresponds to a configuration ($P_1$, $P_2$, $P_3$), while columns represent the scenarios associated with panels (a)--(c). The simulations reveal distinct behaviours consistent with the bifurcation structure. In some cases, trajectories converge monotonically to $\mathbf{U}_3$, indicating strong stability and rapid damping of perturbations (panels (a,d,e,g--i)). In others, sustained oscillations emerge, a clear signature of Hopf bifurcation, where feedback between vegetation recovery and herbivore growth drives cyclic fluctuations (panels (b,c,f)). Ecologically, these patterns confirm the theoretical mechanisms: low mortality combined with high conversion efficiency promotes rapid herbivore growth, intensifying grazing pressure and destabilizing vegetation, which triggers oscillatory dynamics. On the contrary, higher mortality mitigates this effect, allowing vegetation recovery and preventing large-amplitude cycles.

\begin{figure}[t!]
\centering
\includegraphics[width=1\textwidth]{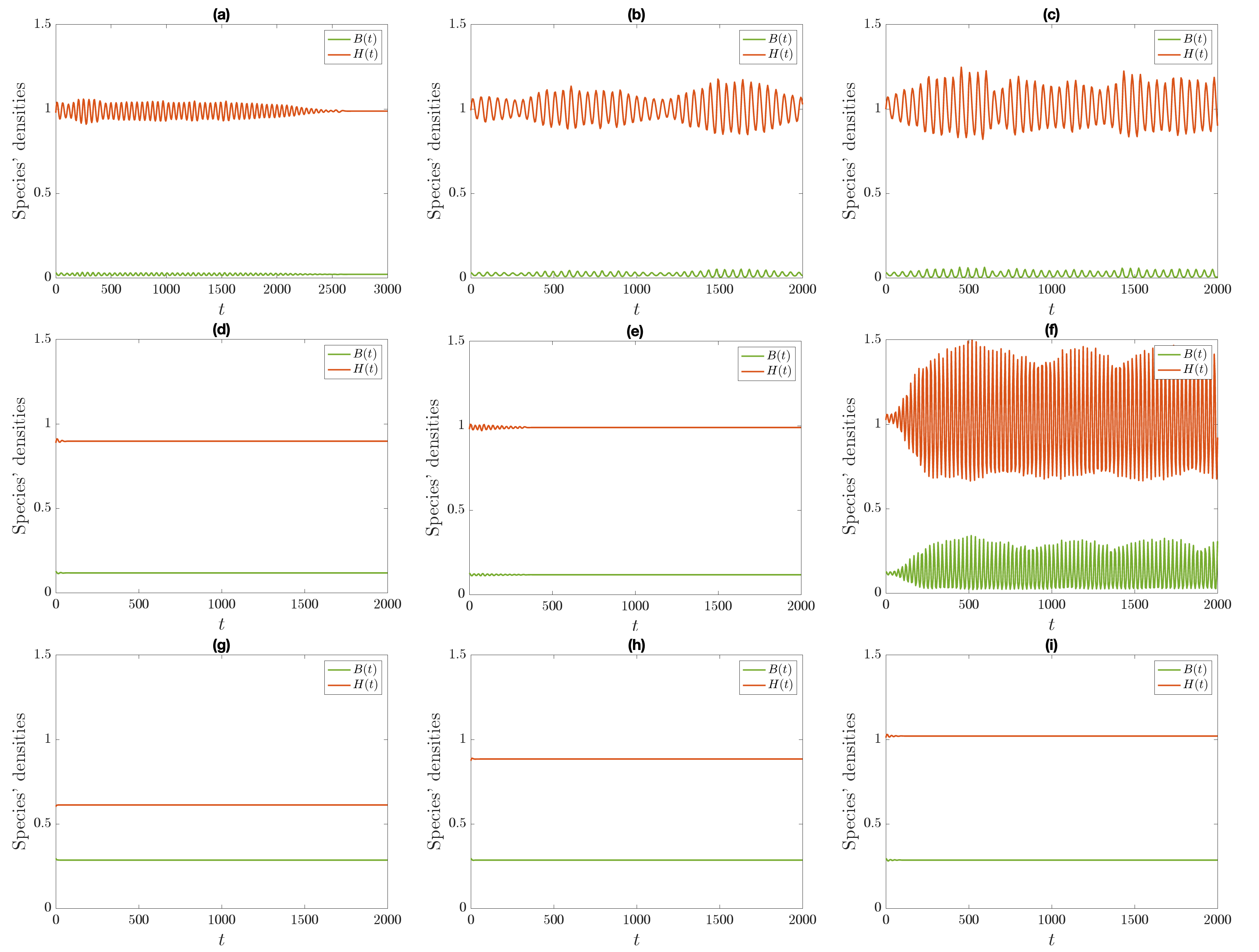}
\caption{Time series of plant biomass and herbivore density after perturbations of $\mathbf{U}_3$ for configurations $P_1$, $P_2$, and $P_3$ (rows) corresponding to panels (a)--(c) in Figure~\ref{fig:weak_regime_nodiff}. Columns show the different carrying capacity scenarios associated with panels (a), (b), and (c) of Figure \ref{fig:weak_regime_nodiff}.}
\label{fig:weak_regime_sim}
\end{figure}

Next, spatial diffusion is incorporated to investigate how temporal dynamics interact with pattern formation, focusing on the onset of Turing instabilities and mixed spatiotemporal behaviours. Figure~\ref{fig:weak_regime_diff} summarizes the bifurcation structure in the $(\theta,d_{HB})$-plane through nine panels, illustrating the influence of $\alpha$, $\mu$, and $m$. In each panel, the red solid curve denotes the existence locus of $\mathbf{U}_3$, the green dash-dotted line marks the Hopf threshold $\theta=\theta_{\text{hopf}}$, and the burgundy dashed line indicates the Turing instability. These curves split the parameter space into four regions with distinct dynamics: 
\begin{itemize}
\item \textbf{Stable uniform state}: right of the Hopf locus and below the Turing curve, where mortality and diffusion suppress oscillations and spatial segregation.
\item \textbf{Pure Turing instability}: right of the Hopf curve and above the Turing threshold, generating stationary vegetation patches and herbivore aggregations via cross-diffusion.
\item \textbf{Pure Hopf instability}: left of the Hopf locus and below the Turing curve, producing temporal oscillations without spatial modulation.
\item \textbf{Mixed-mode dynamics}: above the Turing threshold and left of the Hopf curve, where both instabilities interact, yielding spatio-temporal patterns periodic in space and oscillatory in time. The intersection of Hopf and Turing loci identifies the codimension-two point (Turing-Hopf bifurcation), near which small parameter changes trigger abrupt transitions among these regimes.
\end{itemize}

Systematic trends emerge across panels. In the first row (a)--(c), increasing the herbivore-to-plant diffusion ratio $d$ leaves the Hopf locus unchanged but shifts the Turing threshold upward, reducing pure Turing instability and delaying mixed-mode dynamics. Ecologically, stronger herbivore dispersal mitigates spatial segregation, favouring uniform or oscillatory states. In the second row (d)--(f), higher carrying capacity $\alpha$ enlarges the coexistence domain and shifts the Hopf curve toward higher mortality, broadening oscillatory regions, while the Turing threshold moves downward, slightly expanding spatial instability. In the third row (g)--(i), toxicity $\mu$ mainly reduces coexistence feasibility and enlarges the Hopf region, whereas the Turing locus remains nearly unaffected, indicating weak sensitivity of spatial patterns to chemical defenses. Finally, in the fourth row (l)--(n), increasing conversion efficiency $m$ markedly enlarges coexistence and shifts both Hopf and Turing thresholds, enhancing the potential for oscillatory and patterned dynamics. From an ecological perspective, high $\alpha$ and $m$ promote coexistence but increase the risk of Hopf-driven oscillations, whereas strong toxicity suppresses temporal cycles and favours pattern formation.

\begin{figure}[p]
\centering
\includegraphics[width=1\textwidth]{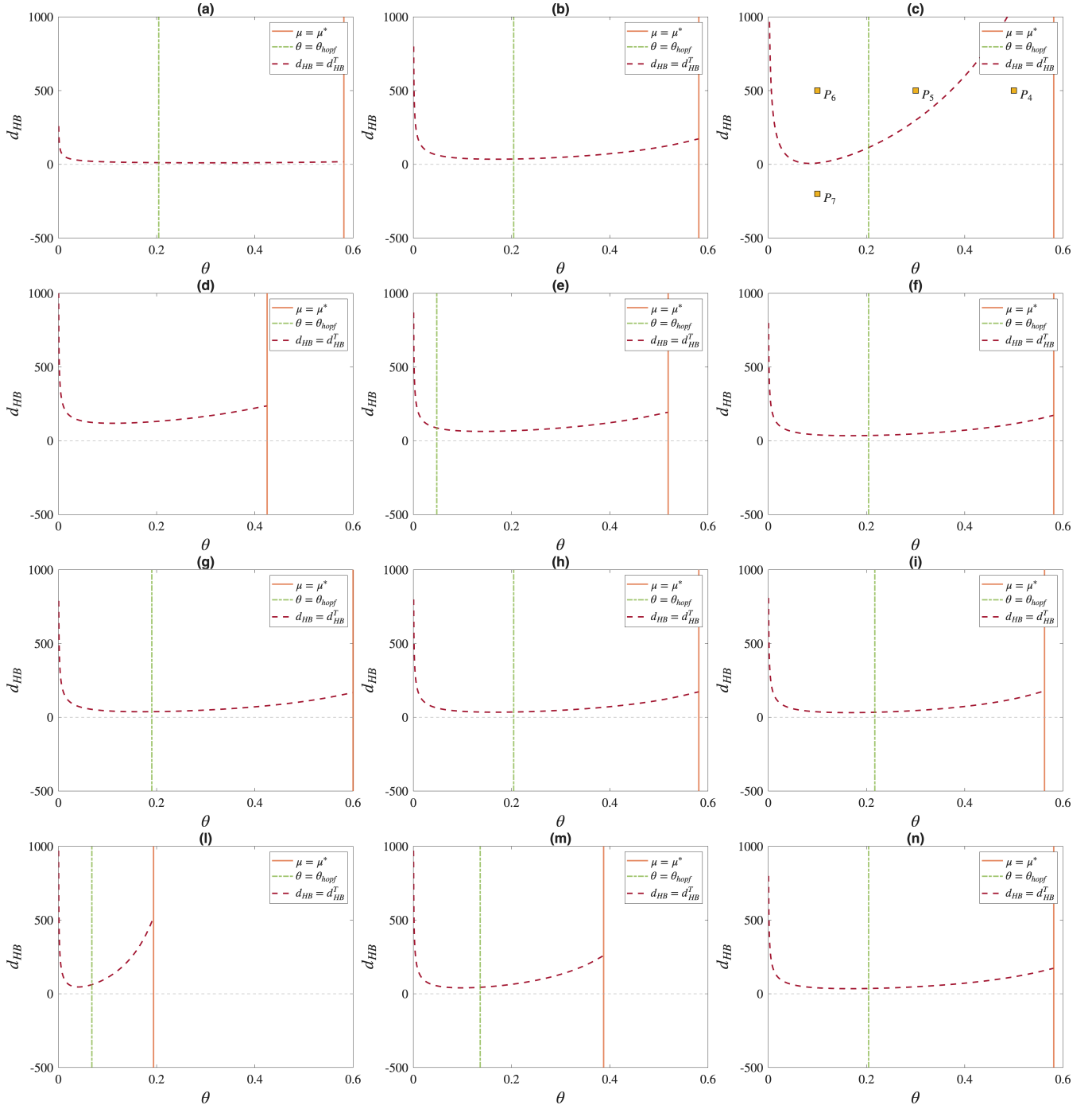}
\caption{Bifurcation diagrams under weak toxicity ($\mu \in [1/4,1/2]$) in the $(\theta,d_{HB})$-plane. Panels: (a)--(c) effect of diffusion ratio $d$; (d)--(f) carrying capacity $\alpha$; (g)--(i) toxicity $\mu$; (l)--(n) conversion efficiency $m$. Red solid line: existence locus; green dash-dotted: Hopf threshold; burgundy dashed: Turing threshold. Yellow squares indicate configurations ($P_4$--$P_7$) used in simulations. The parameter values are $d = 10$ in (a), $d = 100$ in (b, d-n), $d = 1000$ in (c), $\alpha = 0.5$ in (d), $\alpha = 0.75$ in (e), $\alpha = 1$ in (a-c, f, g-n), $\mu = 0.4$ in (g), $\mu = 0.45$ in (a-f, h, l-n), $\mu = 0.5$ in (l), $m = 0.5$ in (l), $m = 1$ in (m) and $m = 1.5$ in (a-i, n).}
\label{fig:weak_regime_diff}
\end{figure}

To corroborate the theoretical predictions at the onset of different dynamical regimes, numerical simulations of system~\eqref{eq:adimmod} are performed under parameter configurations selected from Figure~\ref{fig:weak_regime_diff}(c). Four representative points in the $(\theta,d_{HB})$-plane, denoted by P$_4$, P$_5$, P$_6$, and P$_7$, are considered, each belonging to a distinct region of the bifurcation diagram. Figure \ref{fig:weak_regime_diff_sim} illustrates the spatio-temporal evolution of plant biomass for these configurations. Panel (a) is associated with P$_4$, which lies in the stability region of the coexistence equilibrium $\mathbf{U}_3$ so that the system converges monotonically to a spatially uniform state, confirming that both temporal and spatial instabilities are suppressed.
Differently, the framework considered in panel (b), associated with P$_5$, belongs to the domain of pure Turing instability. The equilibrium $\mathbf{U}_3$ remains stable under homogeneous perturbations but loses stability when spatial heterogeneity is introduced, leading to stationary patterns characterized by vegetation patches and herbivore aggregations. Note that these structures arise from cross-diffusion mechanisms that drive herbivores away from dense vegetation, illustrating how movement strategies shape landscape heterogeneity. Then, the setting shown in panel (c), corresponding to P$_6$, falls within the mixed-mode region where Hopf and Turing instabilities interact. The resulting dynamics combine temporal oscillations with spatial periodicity, producing spatio-temporal patterns that alternate between vegetation recovery and herbivore clustering. Finally, panel (d) is associated with P$_7$, which lies in the region of pure Hopf instability. Here, $\mathbf{U}_3$ loses stability through a Hopf bifurcation, generating homogeneous oscillations in time without spatial modulation. From an ecological viewpoint, this regime corresponds to cyclic fluctuations in vegetation and herbivore densities driven by toxin-mediated feedback under low mortality, while diffusion remains insufficient to destabilize spatial uniformity.

\begin{figure}[b!]
\centering
\includegraphics[width=0.8\textwidth]{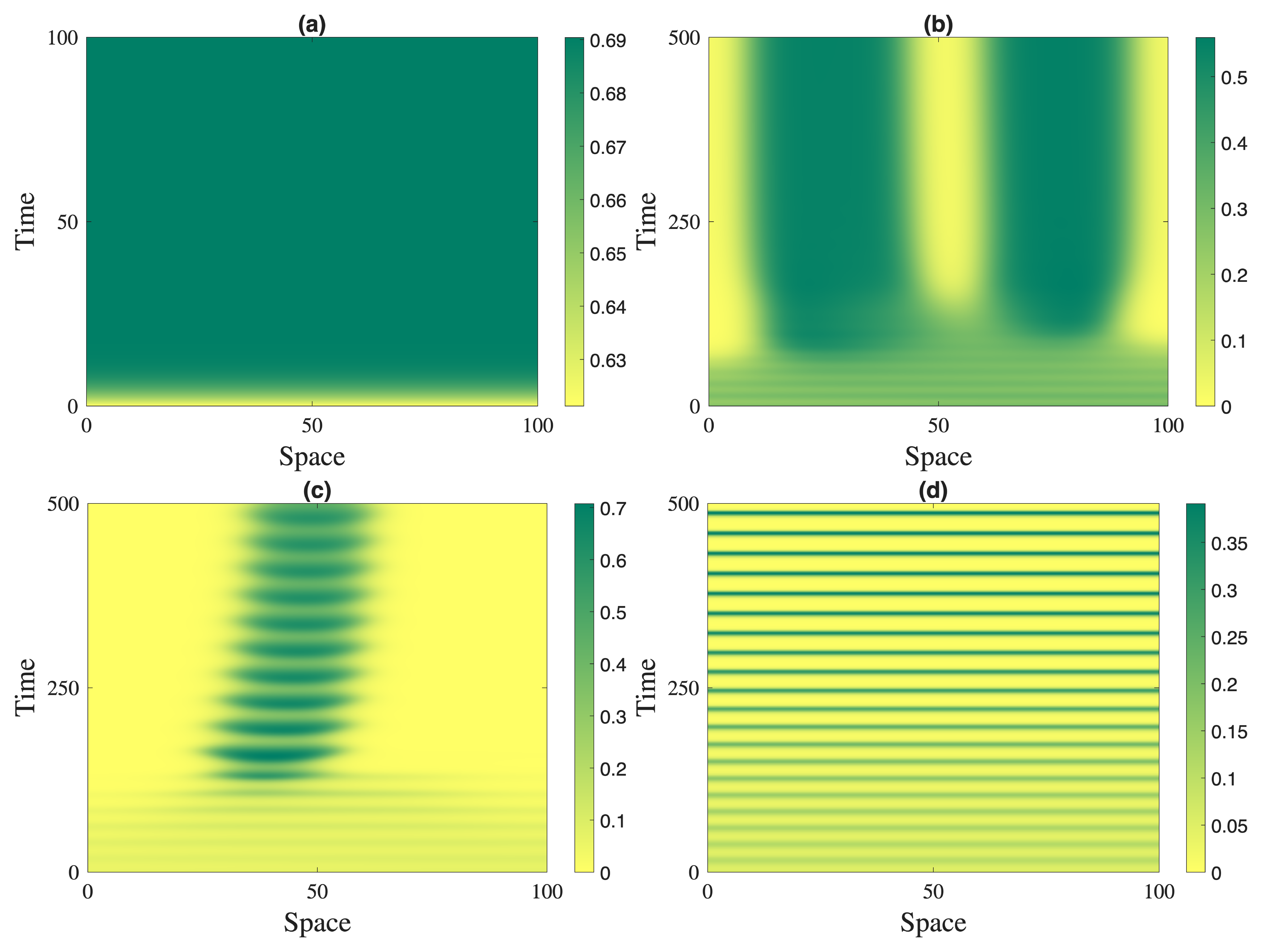}
\caption{Spatio-temporal evolution of plant biomass for configurations $P_4$--$P_7$ marked in Figure~\ref{fig:weak_regime_diff}(c). Panel (a) corresponds to configuration P$_4$, panel (b) to P$_5$, panel (c) to P$_6$, and panel (d) to P$_7$.}
\label{fig:weak_regime_diff_sim}
\end{figure}

Overall, these simulations corroborate the analytical framework, demonstrating how small parameter variations can trigger transitions between uniform states, stationary patterns, and oscillatory structures.

\subsubsection{Strong toxicity regime}

The system \eqref{eq:adimmod} is now examined under strong toxicity, $\mu \in (1/2,1]$, where the toxin-modified functional response becomes unimodal, introducing qualitative differences compared to the weak regime discussed in Section \ref{subsec:weak}. The presence of a maximum in the consumption rate may allow an additional equilibrium $\mathbf{U}_4$, which is always unstable.

Numerical simulations confirm these theoretical predictions and explore oscillatory behaviour through Hopf bifurcation. As in the previous regime, diffusive effects are initially neglected to focus on local interactions. Figure~\ref{fig:strong_regime_nodiff} summarizes the bifurcation structure in the $(\theta,\mu)$-plane for different ecological parameters, highlighting the existence locus of $\mathbf{U}_3$ (red curve) and the Hopf threshold $\theta=\theta_{\text{hopf}}$ (green curve). As in the weak toxicity case, stability of $\mathbf{U}_3$ is achieved only to the right of the Hopf curve, while lower mortality values trigger oscillations. The impact of $\alpha$ can be summarized as follows:
\begin{enumerate}[label=$\alpha_\arabic*$), start=4]
\item \textbf{Low carrying capacity ($\alpha \leq 1/2$)}: the Hopf curve lies outside the strong toxicity range, so $\mathbf{U}_3$ remains stable throughout its domain (panel (a)). Limited resources constrain herbivore growth and prevent cyclic behaviour.
\item \textbf{Intermediate carrying capacity ($1/2<\alpha<2/3$)}: the Hopf threshold enters the strong toxicity regime, creating a narrow instability zone (panel (b)). Increased resource availability strengthens trophic interactions, making oscillations possible under low mortality.
\item \textbf{High carrying capacity ($\alpha \geq 2/3$)}: the instability region expands significantly (panel (c)), favouring large-amplitude oscillations or even complex temporal patterns.
\end{enumerate}

Fixing $\alpha$ in the third scenario ($\alpha > 2/3$), panels (d)--(f) illustrate the influence of conversion efficiency $m$. Increasing $m$ shifts the existence locus markedly to the right, enlarging the coexistence domain, while the Hopf curve moves moderately, expanding the instability region. From the ecological viewpoint, higher $m$ promotes herbivore persistence under broader mortality and toxicity conditions but amplifies grazing pressure, which combined with low mortality destabilizes the system. On the contrary, sufficiently high mortality restores stability despite intensified resource exploitation.

\begin{figure}[t!]
\centering
\includegraphics[width=1\textwidth]{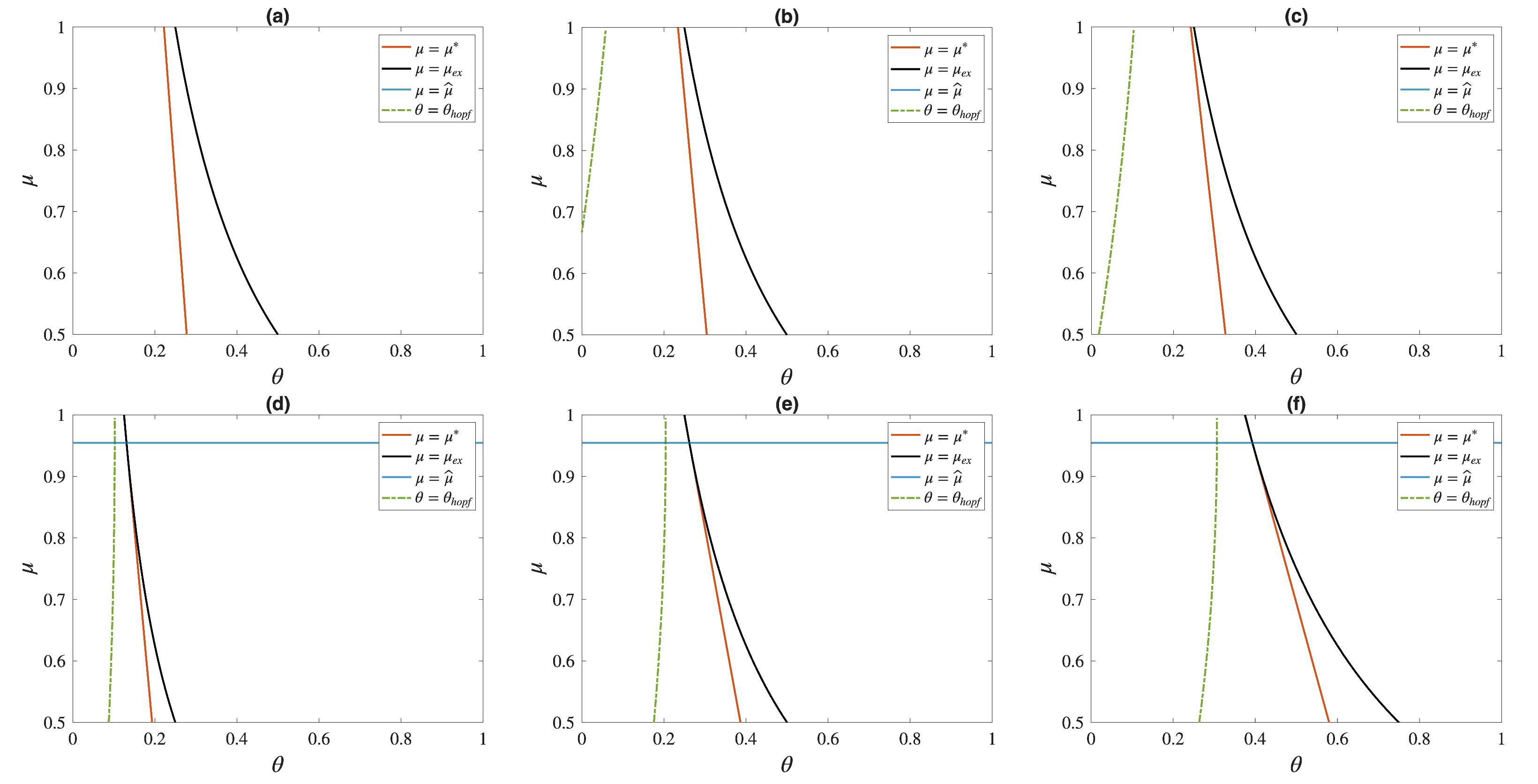}
\caption{Bifurcation diagrams under strong toxicity ($\mu \in (1/2,1]$) without diffusion in the $(\theta,\mu)$-plane. Panels (a)--(c): effect of carrying capacity $\alpha$; panels (d)--(f): effect of conversion efficiency $m$. Red solid line: existence locus of $\mathbf{U}_3$; green dash-dotted line: Hopf threshold. The parameter values are $\alpha = 0.5$ in (a), $\alpha = 0.6$ in (b), $\alpha = 0.7$ in (c), $\alpha = 1.1$ in (d-f), $m = 0.5$ in (d), $m = 1$ in (a-c, e), $m = 1.5$ in (f).}
\label{fig:strong_regime_nodiff}
\end{figure}

%To validate these theoretical predictions, numerical simulations are performed under the parameter configurations highlighted in Figure~\ref{fig:strong_regime_nodiff}(a)--(c). Figure~\ref{fig:strong_regime_sim} reports the time series of plant biomass and herbivore density following small perturbations of the coexistence equilibrium $\mathbf{U}_3$. Each row corresponds to a configuration ($Q_1$, $Q_2$, $Q_3$), while columns represent the scenarios associated with panels (a)--(c) of Figure \ref{fig:strong_regime_nodiff}.
%The simulations reveal distinct behaviours consistent with the bifurcation structure. In some cases, trajectories converge monotonically to $\mathbf{U}_3$, confirming local stability. In others, persistent oscillations emerge, indicating a Hopf bifurcation where feedback between vegetation recovery and herbivore growth drives cyclic fluctuations.

%\begin{figure}[p!]
%\centering
%\includegraphics[width=1\textwidth]{Figure8.eps}
%\caption{Time series of plant biomass and herbivore density after perturbations of $\mathbf{U}_3$ for configurations $Q_1$--$Q_3$ (rows) corresponding to panels (a)--(c) in Figure~\ref{fig:strong_regime_nodiff} (columns).}
%\label{fig:strong_regime_sim}
%\end{figure}

Finally, spatial diffusion is incorporated to investigate how temporal dynamics interact with pattern formation. The bifurcation diagrams in Figure~\ref{fig:strong_regime_diff} show that, despite the transition to strong toxicity, the qualitative influence of ecological parameters remains consistent with the weak regime. Higher carrying capacity $\alpha$ and conversion efficiency $m$ enlarge the coexistence domain and extend regions where Hopf and Turing instabilities occur, while random diffusion $d$ retains its stabilizing role against spatial heterogeneity. Variations in toxicity $\mu$ primarily affect coexistence feasibility and the Hopf threshold, leaving the Turing locus nearly unchanged.
Strong toxicity mainly reduces overall coexistence feasibility without altering the relative impact of other traits.

\begin{figure}[p]
\centering
\includegraphics[width=1\textwidth]{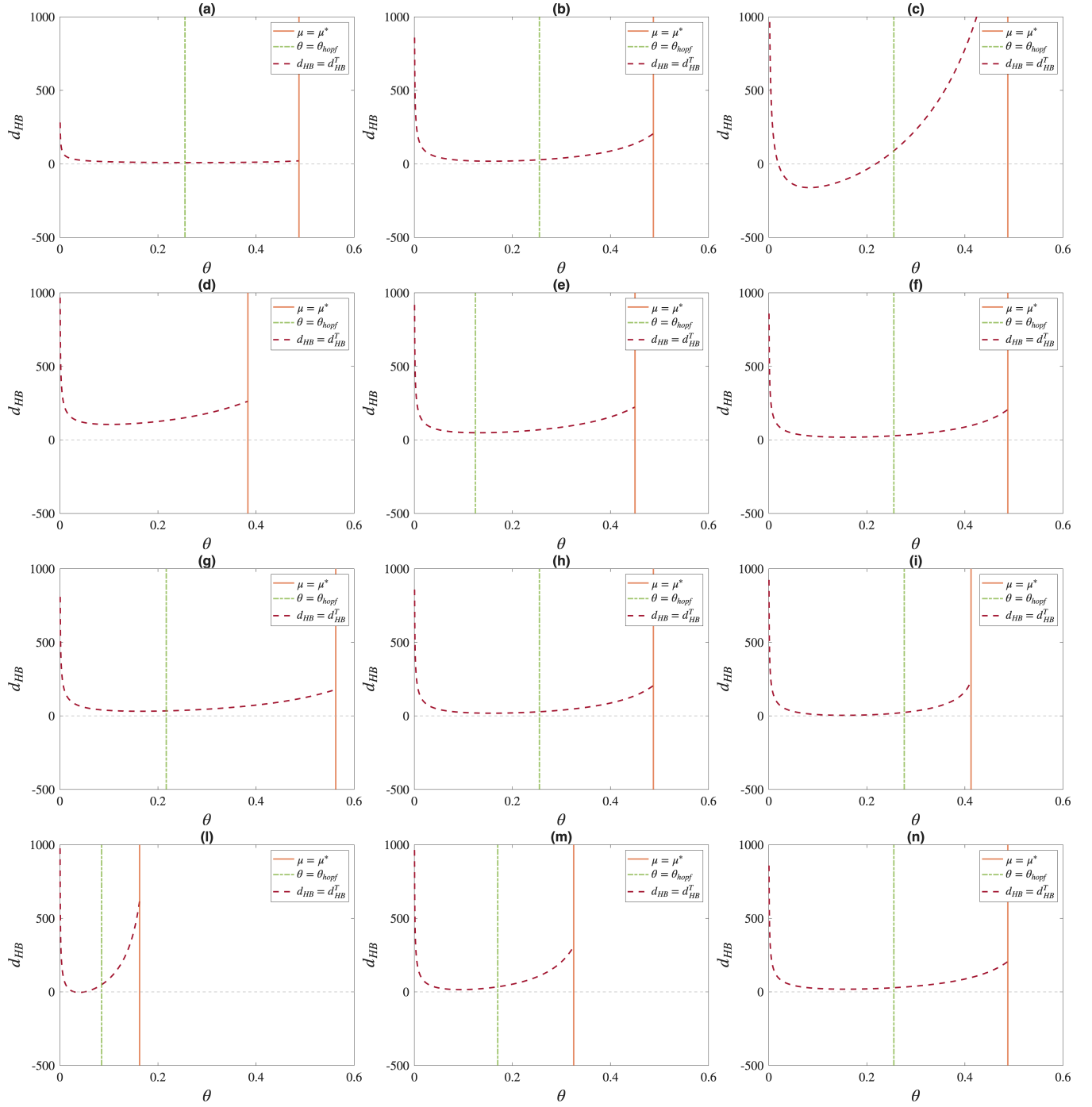}
\caption{Bifurcation diagrams under strong toxicity ($\mu \in (1/2,1]$) in the $(\theta,d_{HB})$-plane. Panels: (a)--(c) effect of diffusion ratio $d$; (d)--(f) carrying capacity $\alpha$; (g)--(i) toxicity $\mu$; (l)--(n) conversion efficiency $m$. Red solid line: existence locus; green dash-dotted: Hopf threshold; burgundy dashed: Turing threshold. Parameter values are the same as those depicted in Figure \ref{fig:weak_regime_diff}. Exception is made for toxicity patameter $\mu$, i.e. $\mu = 0.5$ in (g), $\mu = 0.7$ in (a-f, h, l-n) and $\mu = 0.9$ in (l).}
\label{fig:strong_regime_diff}
\end{figure}

\newpage

\section{Amplitude equations and pattern modulation near instability points}
\label{sec:close}

The existence of spatial patterns is analytically established through a linear stability analysis of the spatially homogeneous steady states, which identifies the conditions for diffusion-driven instabilities. However, this approach cannot fully describe the system near critical points. In particular, linear stability analysis does not predict how the amplitude of emerging patterns evolves beyond the instability nor does it capture the nonlinear interactions. To overcome these limitations, weakly nonlinear technique is employed to characterize the system behavior close to the instability threshold and to derive approximate pattern solutions beyond the linear regime \cite{Ipsen2000, Topaz2010, Casal2011, KUMARI2021}. This framework provides a rigorous basis for understanding the transition from homogeneous states to complex spatio-temporal dynamics. The following sections develop this approach in the one-dimensional setting.

Let us now employ a standard multiscale expansion on both the field vector $\mathbf{U}$ and the bifurcation parameters $(\theta, d_{HB})$ with respect to a small amplitude parameter $0<\varepsilon\ll 1$ as follows

\begin{equation}
\begin{array}{l}
\widetilde{\mathbf{U}} = \varepsilon \widehat{\mathbf{U}}_1 + \varepsilon^2 \widehat{\mathbf{U}}_2 + \varepsilon^3 \widehat{\mathbf{U}}_3 + O(\varepsilon^4),\medskip\\
\theta = \theta_{hopf} + \varepsilon \theta_1 + \varepsilon^2 \theta_2 + \varepsilon^3 \theta_3 + O(\varepsilon^4),\medskip\\
d_{HB} = d_{HB}^T + \varepsilon d_{HB}^{(1)} + \varepsilon^2 d_{HB}^{(2)} + \varepsilon^3 d_{HB}^{(3)} + O(\varepsilon^4),
\end{array}
\label{eq:general_expansion}
\end{equation}
where $\widetilde{\mathbf{U}} = \mathbf{U}-\mathbf{U}_3$.
A hierarchy of slow time scales is also introduced to capture the evolution of the pattern amplitude

\begin{equation}
\partial_t \rightarrow \partial_t + \varepsilon \partial_{T_1} + \varepsilon^2 \partial_{T_2} + \varepsilon^3 \partial_{T_3} + O(\varepsilon^4),
\label{eq:time_scales}
\end{equation}
where $T_j$ ($j=1,2,3$) denote progressively slower temporal scales.

Expanding the nonlinear term $\mathbf{N}(\mathbf{U})$ around the homogeneous steady state $\mathbf{U}_3$ and separating linear, quadratic, and cubic contributions, system~\eqref{eq:modmat}--\eqref{eq:modmat_specifico} becomes

\begin{equation}
\widetilde{\mathbf{U}}_t = \mathcal{L}\,\widetilde{\mathbf{U}} + \frac{1}{2}\mathcal{Q}(\widetilde{\mathbf{U}},\widetilde{\mathbf{U}}) + \frac{1}{6}\mathcal{R}(\widetilde{\mathbf{U}},\widetilde{\mathbf{U}},\widetilde{\mathbf{U}}),
\label{eq:model_oper_hopf}
\end{equation}
being $\mathcal{L}=\mathcal{L}^\ast+\text{M}\partial_{xx}$, and the nonlinear operators $\mathcal{Q}$ and $\mathcal{R}$ are defined by

\begin{equation}
\mathcal{Q}(\mathbf{x},\mathbf{y})=(\mathbf{x}\cdot\nabla)(\mathbf{y}\cdot\nabla)\mathbf{N},\qquad
\mathcal{R}(\mathbf{x},\mathbf{y},\mathbf{z})=(\mathbf{x}\cdot\nabla)(\mathbf{y}\cdot\nabla)(\mathbf{z}\cdot\nabla)\mathbf{N}.
\label{eq:op_Q_R}
\end{equation}
Using expansions~\eqref{eq:general_expansion}-\eqref{eq:time_scales}, the operators involved in \eqref{eq:model_oper_hopf} can be expanded as

\begin{equation}
\begin{array}{l}
\mathcal{L}=\mathcal{L}_{TH}+\varepsilon \left( \theta_1 \partial_\theta \mathcal{L}^\ast + d_{HB}^{(1)} \partial_{d_{HB}} \text{M} \partial_{xx} \right)+\varepsilon^2\!\left(\theta_2 \partial_\theta \mathcal{L}^\ast +\frac{1}{2}\theta_1^2 \partial_{\theta \theta} \mathcal{L}^\ast + d_{HB}^{(2)}\partial_{d_{HB}} \text{M} \partial_{xx}\right)+ \medskip \\
\hspace{1cm} + \varepsilon^3\!\left(\theta_3 \partial_\theta \mathcal{L}^\ast + \theta_1\theta_2 \partial_{\theta \theta}\mathcal{L}^\ast + \frac{1}{6}\theta_1^3 \partial_{\theta \theta \theta} \mathcal{L}^\ast + d_{HB}^{(1)} \partial_{d_{HB}} \text{M} \partial_{xx} \right)+O(\varepsilon^4),\medskip\\
\mathcal{Q}(\widetilde{\mathbf{U}},\widetilde{\mathbf{U}})=\varepsilon^2\mathcal{Q}(\widehat{\mathbf{U}}_1,\widehat{\mathbf{U}}_1)+\varepsilon^3\!\left[2\mathcal{Q}(\widehat{\mathbf{U}}_1, \widehat{\mathbf{U}}_2)+\theta_1\mathcal{Q}_\theta(\widehat{\mathbf{U}}_1,\widehat{\mathbf{U}}_1)\right]+O(\varepsilon^4),\medskip\\
\mathcal{R}(\widetilde{\mathbf{U}},\widetilde{\mathbf{U}},\widetilde{\mathbf{U}})=\varepsilon^3\mathcal{R}(\widehat{\mathbf{U}}_1,\widehat{\mathbf{U}}_1,\widehat{\mathbf{U}}_1)+O(\varepsilon^4),
\end{array}
\label{eq:exp_op_hopf}
\end{equation}
with $\mathcal{L}_{TH} = \mathcal{L}^\ast|_{\theta_{hopf}}+\text{M}|_{d_{HB}^T}\partial_{xx}$.

Collecting terms of identical powers of $\varepsilon$ yields

\begin{equation}
\begin{array}{lll}
\left( \partial_t - \mathcal{L}_{TH} \right) \widehat{\mathbf{U}}_1 = 0 & & \text{order}\, \varepsilon, \bigskip\\
\left( \partial_t - \mathcal{L}_{TH} \right) \widehat{\mathbf{U}}_2 = \mathbf{F}_{TH} & & \text{order}\, \varepsilon^2, \bigskip\\
\left( \partial_t - \mathcal{L}_{TH} \right) \widehat{\mathbf{U}}_3 = \mathbf{G}_{TH} & & \text{order}\, \varepsilon^3 ,
\end{array}
\label{eq:ordini}
\end{equation}
with

\begin{equation}
\begin{aligned}
\mathbf{F}_{TH} = &\tfrac{1}{2}\mathcal{Q}(\widehat{\mathbf{U}}_1,\widehat{\mathbf{U}}_1) - \partial_{T_1} \widehat{\mathbf{U}}_1 + \theta_1\partial_\theta \mathcal{L}^\ast \widehat{\mathbf{U}}_1 + d_{HB}^{(1)} \partial_{d_{HB}} \text{M} \partial_{xx} \widehat{\mathbf{U}}_{1}, \\[6pt]
\mathbf{G}_{TH} = & \mathcal{Q}(\widehat{\mathbf{U}}_1,\widehat{\mathbf{U}}_2) + \tfrac{1}{6}\mathcal{R}(\widehat{\mathbf{U}}_1,\widehat{\mathbf{U}}_1,\widehat{\mathbf{U}}_1) - \partial_{T_1}\widehat{\mathbf{U}}_2 - \partial_{T_2}\widehat{\mathbf{U}}_1 + \theta_1 \partial_\theta \mathcal{L}^\ast \widehat{\mathbf{U}}_2 + \theta_2 \partial_\theta \mathcal{L}^\ast \widehat{\mathbf{U}}_1 + \tfrac{1}{2}\theta_1^2 \partial_{\theta \theta} \mathcal{L}^\ast \widehat{\mathbf{U}}_1 +  \\
&+ \tfrac{1}{2}\theta_1 \partial_\theta \mathcal{Q}(\widehat{\mathbf{U}}_1,\widehat{\mathbf{U}}_1) + d_{HB}^{(1)} \partial_{d_{HB}} \text{M} \partial_{xx} \widehat{\mathbf{U}}_{2} + d_{HB}^{(2)} \partial_{d_{HB}} \text{M} \partial_{xx} \widehat{\mathbf{U}}_{1}.
\end{aligned}
\label{eq:ordini_sor_TH}
\end{equation}

\noindent This general framework will now be specialized to the three scenarios: Hopf bifurcation (temporal oscillations), Turing bifurcation (stationary spatial patterns), and Turing--Hopf codimension-two bifurcation (mixed spatiotemporal structures).

\subsection{Hopf bifurcation and temporal oscillations}
\label{sec:oscillatory}

To characterize the weakly nonlinear behavior of oscillatory homogeneous structures, let us now focus on the emergence of a Hopf bifurcation. Since it involves spatially homogeneous oscillations, diffusion plays no role so that \mbox{$\mathcal{L}_{TH}=\mathcal{L}_H=\mathcal{L}^\ast|_{\theta_{hopf}}$} and $d_{HB}^{(1)} = d_{HB}^{(2)} = d_{HB}^{(3)} = 0$. Under this assumption, the dynamics reduce to temporal modulations governed by the hierarchy

\begin{equation}
\begin{array}{lll}
\left( \partial_t - \mathcal{L}_H \right)\widehat{\mathbf{U}}_1 = 0 & & \text{order } \varepsilon,\medskip\\
\left( \partial_t - \mathcal{L}_H \right)\widehat{\mathbf{U}}_2 = \mathbf{F}_H & & \text{order } \varepsilon^2,\medskip\\
\left( \partial_t - \mathcal{L}_H \right)\widehat{\mathbf{U}}_3 = \mathbf{G}_H & & \text{order } \varepsilon^3,
\end{array}
\label{eq:hopf_hierarchy}
\end{equation}
where $\mathbf{F}_H$ and $\mathbf{G}_H$, given by  \eqref{eq:ordini_sor_TH} with $d_{HB}^{(i)} = 0  \, (i = 1,2,3)$ constraints, collect the quadratic and cubic contributions associated with the Hopf bifurcation.
The leading-order solution is given by

\begin{equation}
\widehat{\mathbf{U}}_1=\Omega(T_1,T_2)\,\boldsymbol{\rho}_H e^{\text{i} \omega_H t} + c.c.,
\end{equation}
where $\Omega(T_1,T_2)$ is the complex slowly varying amplitude and $\boldsymbol{\rho}_H = \left[ - \theta_{hopf} \quad \text{i} \omega_H \right]^T$. Note that, higher-order corrections generate harmonics at $2\omega_H$ and $3\omega_H$. 

At second order, the Fredholm solvability condition $\langle \mathbf{F}_H,\boldsymbol{\psi}_H\rangle=0$, being $\boldsymbol{\psi}_H \in \text{Ker}\{\left( \partial_t - \mathcal{L}_H \right)^\dag\}$, enforces $\theta_1=0$ and $\Omega_{T_1}=0$, so that $\Omega$ varies only on the slower scale $T_2$. The correction reads

\begin{equation}
\widehat{\mathbf{U}}_2= |\Omega|^2 \mathbf{u}_{20} + \Omega^2 \mathbf{u}_{22} e^{2 \text{i} \omega_H t} + \overline{\Omega}^2 \overline{\mathbf{u}}_{22} e^{-2 \text{i} \omega_H t},
\end{equation}
where $\mathbf{u}_{20}$ and $\mathbf{u}_{22}$ solve the associated linear problems

\begin{equation}
\mathcal{L}_H \mathbf{u}_{20} = - \mathcal{Q}(\boldsymbol{\rho}_H, \overline{\boldsymbol{\rho}}_H), \qquad
\left( \mathcal{L}_H - 2 \text{i} \omega_H I \right) \mathbf{u}_{22} = - \frac{1}{2} \mathcal{Q}(\boldsymbol{\rho}_H, \boldsymbol{\rho}_H)
\label{eq:u2_definition}
\end{equation}
and ``$\overline{{\color{white}a}}$" denote the complex and conjugate.

Finally, from \eqref{eq:hopf_hierarchy}$_3$, the solvability condition $\langle \mathbf{G}_H,\boldsymbol{\psi}_H\rangle=0$ yields the real Stuart--Landau amplitude equation

\begin{equation}
\frac{\partial \Omega}{\partial T_2}=\sigma_H\Omega-L_H \Omega |\Omega|^2,
\label{eq:ST_hopf}
\end{equation}
where the coefficients $\sigma_H$ and $L_H$ are given by 

\begin{equation}
\sigma_H = \theta_2 \frac{\partial_\theta \mathcal{L}_H \boldsymbol{\rho}_H\cdot \overline{\boldsymbol{\psi}}_H}{\boldsymbol{\rho}_H\cdot\overline{\boldsymbol{\psi}}_H}, 
\qquad L_H = - \frac{\left[\mathcal{Q}(\boldsymbol{\rho}_H, \mathbf{u}_{20}) + \mathcal{Q}(\overline{\boldsymbol{\rho}}_H, \mathbf{u}_{22}) + \frac{1}{2} \mathcal{R} (\boldsymbol{\rho}_H, \boldsymbol{\rho}_H, \overline{\boldsymbol{\rho}}_H) \right]\cdot\overline{\boldsymbol{\psi}}_H}{\boldsymbol{\rho}_H\cdot\overline{\boldsymbol{\psi}}_H}
\end{equation}
and the sign of $\text{Re}\{L_H\}$ determines the bifurcation type. More precisely, for $\text{Re}\{L_H\}>0$, the Hopf bifurcation is supercritical and the amplitude saturates at

\begin{equation}
\Omega_{H_\infty}=\sqrt{\frac{\sigma_H}{L_H}},
\end{equation}
leading to the approximate solution

\begin{equation}
\mathbf{U}(t)=\mathbf{U}_3+\varepsilon \left[ \boldsymbol{\rho}_H \Omega_{H_\infty} e^{\text{i} \omega_H t}+ \overline{\boldsymbol{\rho}}_H \overline{\Omega}_{H_\infty} e^{-\text{i} \omega_H t}\right] + \varepsilon^2 \left[ |\Omega_{H_\infty}|^2 \mathbf{u}_{20} + \Omega_{H_\infty}^2 \mathbf{u}_{22} e^{2 \text{i} \omega_H t} + \overline{\Omega}_{H_\infty}^2 \overline{\mathbf{u}}_{22} e^{-2 \text{i} \omega_H t}\right]+O(\varepsilon^3).
\label{eq:approx}
\end{equation}
On the contrary, when $\text{Re}\{L_H\}<0$ the Hopf bifurcation is subcritical and the amplitude equation predicts an unstable branch near the threshold. In this case, small-amplitude oscillations cannot persist, and the system may exhibit bistability between the trivial state and large-amplitude limit cycles. To approximate these finite-amplitude solutions, the weakly nonlinear expansion must be extended beyond third order, incorporating quintic terms in the amplitude equation. This higher-order analysis captures the saturation mechanism and provides insight into hysteresis phenomena.

\begin{remark}[Ecological interpretation]
The sign of $\text{Re}\{L_H\}$ carries ecological meaning. A supercritical Hopf bifurcation ($\text{Re}\{L_H\}>0$) corresponds to smooth, small-amplitude cycles reflecting resilience, where vegetation and herbivore densities fluctuate predictably under self-regulating mechanisms. Conversely, a subcritical Hopf bifurcation ($\text{Re}\{L_H\}<0$) signals fragility: oscillations may arise abruptly with large amplitude, inducing bistability and ecological tipping points. Such dynamics amplify grazing pressure and vegetation stress, reducing ecosystem resilience and increasing susceptibility to disturbances.
\end{remark}

\subsection{Turing bifurcation and spatial heterogeneity}
\label{sec:periodic_patterns}

While the Hopf analysis addressed temporal oscillations, the Turing mechanism requires a complementary viewpoint. To capture the slow modulation of these structures, a weakly nonlinear multiple-scale expansion is employed, analogous in spirit to the Hopf case but tailored to spatial heterogeneity.

Since Turing instability involves spatially nonuniform perturbations without temporal oscillations, the general multiscale framework outlined in \eqref{eq:general_expansion}-\eqref{eq:time_scales} is specialized by neglecting the fast time derivative $\partial_t$ and $\theta_1 = \theta_2 = \theta_3 = 0$. Under these assumptions, the equations \eqref{eq:ordini} reduce to

\begin{equation}
\begin{array}{lll}
\mathcal{L}_T \widehat{\mathbf{U}}_1 = 0 & & \text{order } \varepsilon,\medskip\\
\mathcal{L}_T \widehat{\mathbf{U}}_2 = \mathbf{F}_T & & \text{order } \varepsilon^2,\medskip\\
\mathcal{L}_T \widehat{\mathbf{U}}_3 = \mathbf{G}_T & & \text{order } \varepsilon^3,
\end{array}
\label{eq:turing_hierarchy}
\end{equation}
being $\mathcal{L}_T = \mathcal{L}^\ast + \text{M}|_{d_{HB}^T}\partial_{xx}$, $\mathbf{F}_T$ and $\mathbf{G}_T$ given by \eqref{eq:ordini_sor_TH} with $\theta_i = 0 \, (i = 1,2,3)$. At leading order, the nullspace of $\mathcal{L}_T$ selects the excited wavenumber $\hat{k}_T$, yielding

\begin{equation}
\widehat{\mathbf{U}}_1 = \Gamma(T_1,T_2)\,\boldsymbol{\rho}_T e^{\text{i} \hat{k}_T x} + c.c.,
\label{eq:Turing_W1}
\end{equation}
where $\Gamma$ is the slowly varying amplitude and $\boldsymbol{\rho}_T = [ \theta \quad \theta H'(B_3)-\hat{k}_T^2 ]^T$ the eigenvector of $\mathcal{L}_T$. Applying the Fredholm alternative at second order enforces $\Gamma_{T_1}=0$ and $d_{HB}^{(1)}=0$, so that $\Gamma$ depends only on $T_2$. The correction reads

\begin{equation}
\widehat{\mathbf{U}}_2 = |\Gamma|^2 \mathbf{w}_{20} + \mathbf{w}_{22} \left(\Gamma^2 e^{2 \text{i} \hat{k}_T x} + \overline{\Gamma}^2 e^{-2 \text{i} \hat{k}_T x} \right),
\label{eq:Turing_W2}
\end{equation}
where $\mathbf{w}_{20}$ and $\mathbf{w}_{22}$ solve

\begin{equation}
\mathcal{L}^\ast \mathbf{w}_{20} = - \mathcal{Q}(\boldsymbol{\rho}_T,\boldsymbol{\rho}_T),\qquad
\left(\mathcal{L}^\ast-4\hat{k}_T^2\text{M}\vert_{d_{HB}^T}\right)\mathbf{w}_{22}=-\tfrac{1}{2}\mathcal{Q}(\boldsymbol{\rho}_T,\boldsymbol{\rho}_T).
\label{eq:w20_w22}
\end{equation}

At third order, the solvability condition $\langle \mathbf{G}_T,\boldsymbol{\psi}_T\rangle=0$, being $\boldsymbol{\psi}_T \in \text{Ker} \{{\mathcal{L}_T}^\dag \}$, yields the real Stuart--Landau amplitude equation

\begin{equation}
\frac{\partial \Gamma}{\partial T_2}=\sigma_T\Gamma-L_T\Gamma |\Gamma|^2,
\label{eq:Turing_SL}
\end{equation}
where

\begin{equation}
\sigma_T=-\hat{k}^2_Td_{HB}^{(2)}\frac{ \partial_{d_{HB}} \text{M} \boldsymbol{\rho}_T\cdot\boldsymbol{\psi}_T}{\boldsymbol{\rho}_T\cdot\boldsymbol{\psi}_T},\quad
L_T=-\frac{\left[\mathcal{Q}(\boldsymbol{\rho}_T,\mathbf{w}_{20}+\mathbf{w}_{22})+\frac{1}{2}\mathcal{R}(\boldsymbol{\rho}_T,\boldsymbol{\rho}_T,\boldsymbol{\rho}_T)\right]\cdot\boldsymbol{\psi}_T}{\boldsymbol{\rho}_T\cdot\boldsymbol{\psi}_T}.
\end{equation}

Since $\sigma_T>0$ once the onset of instability is crossed, the sign of $L_T$ dictates the bifurcation type. For $L_T > 0$, the branch emerges smoothly (supercritical), and the asymptotic amplitude is

\begin{equation}
\Gamma_{T_\infty}=\sqrt{\frac{\sigma_T}{L_T}},
\end{equation}
leading to the approximate stationary pattern

\begin{equation}
\mathbf{U}(x)=\mathbf{U}_3+\varepsilon \boldsymbol{\rho}_T \left[ \Gamma_{T_\infty} e^{ \text{i} \hat{k}_T x} + \overline{\Gamma}_{T_\infty} e^{-\text{i} \hat{k}_T x} \right] +\varepsilon^2 \left[ |\Gamma_{T_\infty}|^2 \mathbf{w}_{20}+ \mathbf{w}_{22}  \left(\Gamma_{T_\infty}^2 e^{2 \text{i} \hat{k}_T x} + \overline{\Gamma}_{T_\infty}^2 e^{-2 \text{i} \hat{k}_T x} \right)\right]+O(\varepsilon^3).
\end{equation}
Also in this case, when the bifurcation is subcritical ($L_T < 0$), the cubic approximation provided by the Stuart--Landau equation is no longer sufficient to capture the amplitude saturation. In such cases, higher-order terms must be included in the weakly nonlinear expansion, as they govern the emergence of finite-amplitude patterns and possible hysteresis phenomena.

\begin{remark}[Ecological implications of spatial heterogeneity]
Unlike temporal oscillations, Turing patterns directly affect landscape structure. Smooth, supercritical bifurcations generate moderate vegetation contrasts, promoting gradual adaptation and resilience. Conversely, subcritical transitions produce abrupt, large-amplitude segregation, often accompanied by hysteresis. Such sharp spatial discontinuities indicate weak ecosystems where small environmental changes can trigger sudden shifts in resource distribution and herbivore clustering.
\end{remark}

\subsection{Turing-Hopf bifurcation, temporal oscillations and spatial heterogeneity}
\label{sec:TH}

To characterize the weakly nonlinear behavior of mixed spatiotemporal structures emerging at a Turing--Hopf bifurcation, the analysis now focuses on the \emph{codimension-two point} in the parameter space $(\theta_{hopf}, d_{HB}^T)$. In this neighborhood, the bifurcation is crossed along the intersection of Hopf and Turing loci, and therefore both bifurcation parameters vary according to \eqref{eq:general_expansion}, so that neither effect is suppressed. Under these assumptions, the dynamics involve coupled temporal and spatial modulations governed by the systems \eqref{eq:ordini}.

At leading order, the nullspace of $\mathcal{L}_{TH}$ selects two critical modes, the Hopf frequency $\omega_H$ and the Turing wavenumber $\hat{k}_T$, so that the first-order solution reads

\begin{equation}
\widehat{\mathbf{U}}_1 = \Omega(T_1,T_2)\,\boldsymbol{\rho}_H e^{i\omega_H t} + \Gamma(T_1,T_2)\,\boldsymbol{\rho}_T e^{\text{i} \hat{k}_T x} + c.c.,
\label{eq:TH_U1}
\end{equation}
where $\Omega$ and $\Gamma$ are complex amplitudes associated with temporal oscillations and spatial patterns, respectively. Note that higher-order corrections introduce mixed harmonics which characterize spatiotemporal modulation.

At second order, the solvability conditions enforce $\Omega_{T_1}=\Gamma_{T_1}=\theta_1=d_{HB}^{(1)}=0$, so that both amplitudes evolve only on the slower scale $T_2$. The second-order correction reads

\begin{equation}
\begin{aligned}
\widehat{\mathbf{U}}_2= & |\Omega|^2 \mathbf{u}_{20} + \Omega^2 \mathbf{u}_{22} e^{2 \text{i} \omega_H t} + \overline{\Omega}^2 \overline{\mathbf{u}}_{22} e^{-2 \text{i} \omega_H t} + |\Gamma|^2 \mathbf{w}_{20} + \Gamma^2 \mathbf{w}_{22} e^{2 \text{i} \hat{k}_T x} + \overline{\Gamma}^2 \overline{\mathbf{w}}_{22} e^{-2 \text{i} \hat{k}_T x} + \medskip \\
& + \mathbf{z}_{20} \left\{ \Omega \Gamma e^{\text{i} \left( \omega_H t + \hat{k}_T x \right)} + \Omega \overline{\Gamma} e^{\text{i} \left( \omega_H t - \hat{k}_T x \right)} \right\}+ \overline{\mathbf{z}}_{20} \left\{ \overline{\Omega} \overline{\Gamma} e^{-\text{i} \left( \omega_H t + \hat{k}_T x \right)} + \overline{\Omega} \Gamma e^{-\text{i} \left( \omega_H t - \hat{k}_T x \right)} \right\},
\end{aligned}
\label{eq:TH_U2}
\end{equation}
with $\mathbf{u}_{20}$, $\mathbf{u}_{22}$, $\mathbf{w}_{20}$, $\mathbf{w}_{22}$ are defined by \eqref{eq:u2_definition}-\eqref{eq:w20_w22}, whereas $\mathbf{z}_{20}$ solve

\begin{equation}
\left[\mathcal{L}_H - \text{i} \omega_H I - \hat{k}_T^2 \text{M} \right] \mathbf{z}_{20} = - \mathcal{Q}(\boldsymbol{\rho}_H,\boldsymbol{\rho}_T).
\label{eq:w20_w22}
\end{equation}

Finally, at the third order, the Fredholm alternative solvability condition $\langle \mathbf{G}_{TH}, \psi_{TH} \rangle = 0$ yields the coupled Stuart--Landau amplitude equations

\begin{equation}
\begin{cases}
\frac{\partial \Omega}{\partial T_2} = \sigma_H \Omega - L_H \Omega |\Omega|^2 + \nu_1 \Omega |\Gamma|^2,\medskip \\
\frac{\partial \Gamma}{\partial T_2} = \sigma_{TH} \Gamma - L_T \Gamma |\Gamma|^2 + \nu_2 \Gamma |\Omega|^2,
\end{cases}
\label{eq:TH_SL}
\end{equation}
where 

\begin{equation}
\begin{array}{c}
\sigma_{TH} = \sigma_T + \theta_2 \frac{\partial_\theta \mathcal{L}^\ast \boldsymbol{\rho}_T\cdot\boldsymbol{\psi}_T}{\boldsymbol{\rho}_T\cdot\boldsymbol{\psi}_T}, \medskip \\
\begin{array}{l l}
\nu_1 = \frac{\left[ \mathcal{Q}(\boldsymbol{\rho}_H, \mathbf{w}_{20}) + 2 \mathcal{Q}(\boldsymbol{\rho}_T, \mathbf{z}_{20}) + \mathcal{R}(\boldsymbol{\rho}_H,\boldsymbol{\rho}_T,\boldsymbol{\rho}_T) \right]\cdot\overline{\boldsymbol{\psi}}_H}{\boldsymbol{\rho}_H\cdot\overline{\boldsymbol{\psi}}_H}, &
\nu_2 = \frac{\left[ \mathcal{Q}(\boldsymbol{\rho}_T, \mathbf{u}_{20}) + \mathcal{Q}(\boldsymbol{\rho}_H, \overline{\mathbf{z}}_{20}) + \mathcal{Q}(\overline{\boldsymbol{\rho}}_H, \mathbf{z}_{20}) + \mathcal{R}(\boldsymbol{\rho}_T,\boldsymbol{\rho}_H,\overline{\boldsymbol{\rho}}_H) \right]\cdot\boldsymbol{\psi}_T}{\boldsymbol{\rho}_T\cdot\boldsymbol{\psi}_T}.
\end{array}
\end{array}
\end{equation}
Note that, after the bifurcation threshold, $\text{Re}\{\sigma_H\},\sigma_T>0$ measure the distance from the Hopf and Turing thresholds while $\nu_1$ and $\nu_2$ quantify the cross-interaction between modes and govern the persistence of mixed spatiotemporal structures. The quantities $L_H$ and $L_T$, as well as $\nu_1$ and $\nu_2$, determine whether the bifurcations are supercritical or subcritical.

Finally, let us examine the stationary points of the amplitude equations~\eqref{eq:TH_SL} to provide insight into the possible long-term configurations of the emerging patterns. System~\eqref{eq:TH_SL} admits the following stationary configurations $(\Omega,\Gamma)$:
\begin{itemize}
    \item \textbf{Trivial state:} $(0,0)$, corresponding to the uniform equilibrium $\mathbf{U}_3$. 
    \newline \emph{Mathematically:} This state is stable when both growth rates $\text{Re}\{\sigma_H\},\sigma_{TH}$ are negative and loses stability as soon as one of them becomes positive, crossing the corresponding bifurcation threshold. In this regime, the system remains in its homogeneous configuration.
    \newline \emph{Ecologically:} No oscillations or spatial patterns occur; vegetation and herbivore densities stay uniform, reflecting a resilient regime far from bifurcation thresholds.
    
    \item \textbf{Single-mode branches:}
    \begin{itemize}
        \item Pure Hopf oscillations: $(\Omega_\infty,0)$, with $\Omega_\infty = \sqrt{\sigma_H / L_H}$ and $\Gamma = 0$. 
        \newline \emph{Mathematically:} This branch exists when $\text{Re}\{L_H\}>0$ so that temporal instability dominates while spatial perturbations decay. 
        \newline \emph{Ecologically:} The system exhibits uniform oscillations in time without spatial segregation, representing cyclic fluctuations in plant and herbivore densities.
        
        \item Pure Turing patterns: $(0,\Gamma_\infty)$, with $\Gamma_\infty = \sqrt{\sigma_{TH} / L_T}$ and $\Omega = 0$. 
        \newline \emph{Mathematically:} Feasible when $L_T >0$ so that spatial instability prevails while temporal oscillations are suppressed. 
        \newline \emph{Ecologically:} Stationary vegetation patches and herbivore aggregations emerge, reflecting heterogeneous landscapes under strong spatial feedback.
    \end{itemize}
    
    \item \textbf{Mixed-mode branches:} $(\Omega_{{TH}_\infty},\Gamma_{{TH}_\infty})$, where 
    
    \begin{equation}
    \Omega_{{TH}_\infty} = \sqrt{\frac{\sigma_H L_T + \sigma_{TH} \nu_1}{L_H L_T - \nu_1 \nu_2}}, \qquad
    \Gamma_{{TH}_\infty}= \sqrt{\frac{\sigma_{TH} L_H + \sigma_H \nu_2}{L_H L_T - \nu_1 \nu_2}}.    
    \end{equation}
    \emph{Mathematically:} This coexistence state arises near the codimension-two point and cross-interaction terms $\nu_1,\nu_2$ regulate amplitude saturation. 
    \newline \emph{Ecologically:} Temporal oscillations and spatial patterns coexist, producing complex spatiotemporal structures such as oscillating stripes or breathing patches, often associated with reduced ecosystem resilience.
\end{itemize}

\begin{remark}[Ecological interpretation]
Turing--Hopf dynamics correspond to ecosystems where vegetation and herbivore densities fluctuate both in space and time. Such mixed regimes often signal reduced resilience: patches periodically expand and contract while herbivore clusters oscillate, reflecting strong feedbacks between resource distribution, chemical defenses, and movement strategies.
\end{remark}

\subsection{Numerical investigations}

This section illustrates how numerical simulations corroborate the weakly nonlinear analysis developed in Sections \ref{sec:oscillatory}-\ref{sec:TH}. The goal is to verify the predictive capability of the amplitude equations near Hopf, Turing, and Turing--Hopf bifurcation thresholds. By exploring parameter regimes close to critical loci, the simulations reveal the emergence and modulation of coherent structures, providing a quantitative comparison between theoretical approximations and fully nonlinear dynamics.

Figure~\ref{fig:hopf_amp} provides an integrated view of the comparison between theoretical predictions and numerical results near the Hopf threshold. The analysis begins with panel~(a), which depicts the bifurcation diagram in the $(\theta,d_{HB})$-plane obtained from linear stability theory for the parameter set corresponding to Figure~\ref{fig:strong_regime_diff}(c). This diagram identifies the critical loci where instabilities arise; in particular, the gray dashed line marks $d_{HB}=50$, a reference value used for the continuation analysis discussed next. Starting from this setting, panel~(b) shows the numerical bifurcation diagram computed with the \texttt{pde2path} software \cite{Uecker2014}, where solid lines denote the stable branches and dashed lines the unstable ones. The black curve represents the conductive branch of stationary solutions, which loses stability at the first Hopf bifurcation. From this point, two distinct families of solutions emerge: the blue branch, originating from Hopf bifurcation, and the red branch, associated with Turing instability. Hopf points are highlighted by blue diamonds, Turing points by red circles, and other primary or secondary bifurcations by green squares, revealing the richness of the bifurcation structure beyond the instability. Panel~(c) complements this picture by presenting the theoretical bifurcation diagram derived from the weakly nonlinear analysis in Section~\ref{sec:oscillatory}. Unlike the numerical continuation, which captures multiple branches and secondary bifurcations, the analytical approach focuses exclusively on the first Hopf and Turing thresholds and the initial branch of periodic solutions. This contrast underscores the local nature of the asymptotic expansion and anticipates the discussion on its validity range. Such validity is quantified in panel~(d), which reports the normalized distance from the bifurcation threshold through the parameter $\varepsilon^2$. As expected, the approximation is highly accurate for $\varepsilon^2 \approx 0$, while deviations become significant for $\varepsilon^2 \sim 10^{-2}$ and grow as the system moves away from criticality \cite{Consolo2019, Consolo2024, Consolo2024II}. This observation provides a natural transition to panels~(e)--(h), where the quality of the approximation is assessed in detail. In particular, panels~(e)--(f) compare numerical and analytical solutions for the parameter set indicated by the dotted line in panels~(b)--(c), corresponding to a configuration close to the Hopf threshold. Here, the agreement is excellent, confirming the predictive capability of the Stuart--Landau equation in the weakly nonlinear regime. In contrast, panels~(g)--(h) refer to the dashed-dotted line in panels~(b)--(c), where the system lies farther from the bifurcation point. In this case, the approximation deteriorates, as higher-order terms neglected in the expansion become relevant.

\begin{figure}[t!]
\centering
\includegraphics[width=1\textwidth]{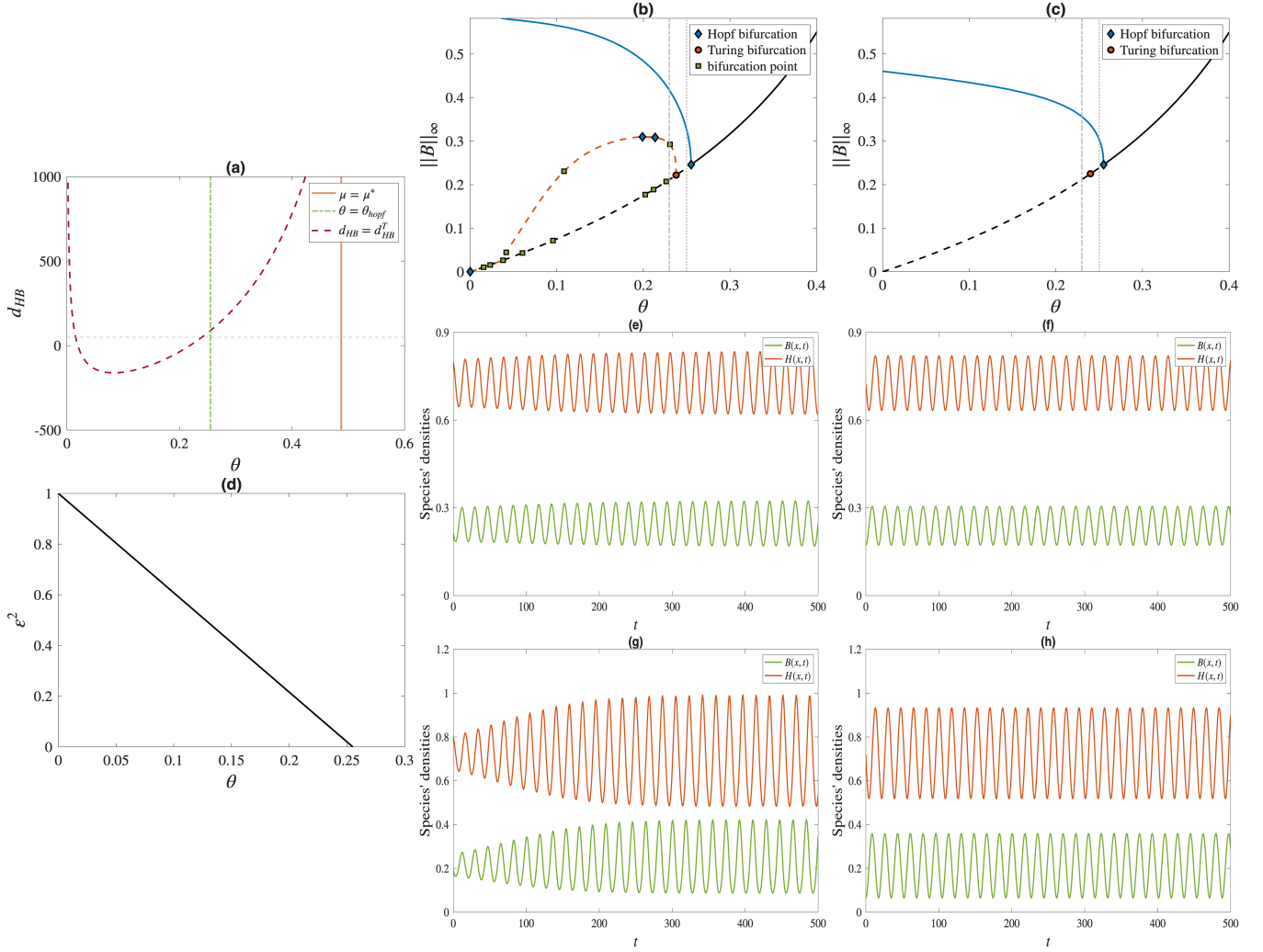}
\caption{(a) Bifurcation diagram in the $(\theta,d_{HB})$-plane for the parameter set of Figure~\ref{fig:strong_regime_diff}(c). 
(b)--(c) Numerical and theoretical bifurcation diagrams: solid lines denote stable branches, dashed lines unstable ones; blue and red curves correspond to Hopf and Turing branches, respectively; the black curve represents the steady state, and markers indicate primary and secondary bifurcations. 
(d) Normalized distance from the bifurcation threshold expressed by $\varepsilon^2$. 
(e)--(f) Numerical and analytical solutions for the parameter set indicated by the dotted line in panels (b)--(c), close to the Hopf threshold. 
(g)--(h) Numerical and analytical solutions for the parameter set indicated by the dashed-dotted line in panels (b)--(c), farther from the Hopf threshold.}
\label{fig:hopf_amp}
\end{figure}

Having clarified the behavior near the Hopf threshold, attention is now turned to the onset of spatial heterogeneity driven by Turing instability. The following analysis mirrors the structure adopted for the Hopf case, combining linear stability results, numerical continuation, and weakly nonlinear predictions to assess the validity of the amplitude equations and their ability to capture pattern morphology beyond the bifurcation point.
Figure~\ref{fig:turing_amp} illustrates the outcomes of this analysis for the Turing case. In detail, panel~(a) identifies the region of instability in the $(\theta,d_{HB})$-plane, while panels~(b)--(c) contrast the richness of the numerical bifurcation structure with the local character of the weakly nonlinear approximation. The normalized distance $\varepsilon^2$ in panel~(d) confirms that the asymptotic description remains accurate only in a narrow neighborhood of the threshold. This trend is reflected in panels~(e)--(f), where analytical and numerical patterns agree closely near criticality.

\begin{figure}[b!]
\centering
\includegraphics[width=1\textwidth]{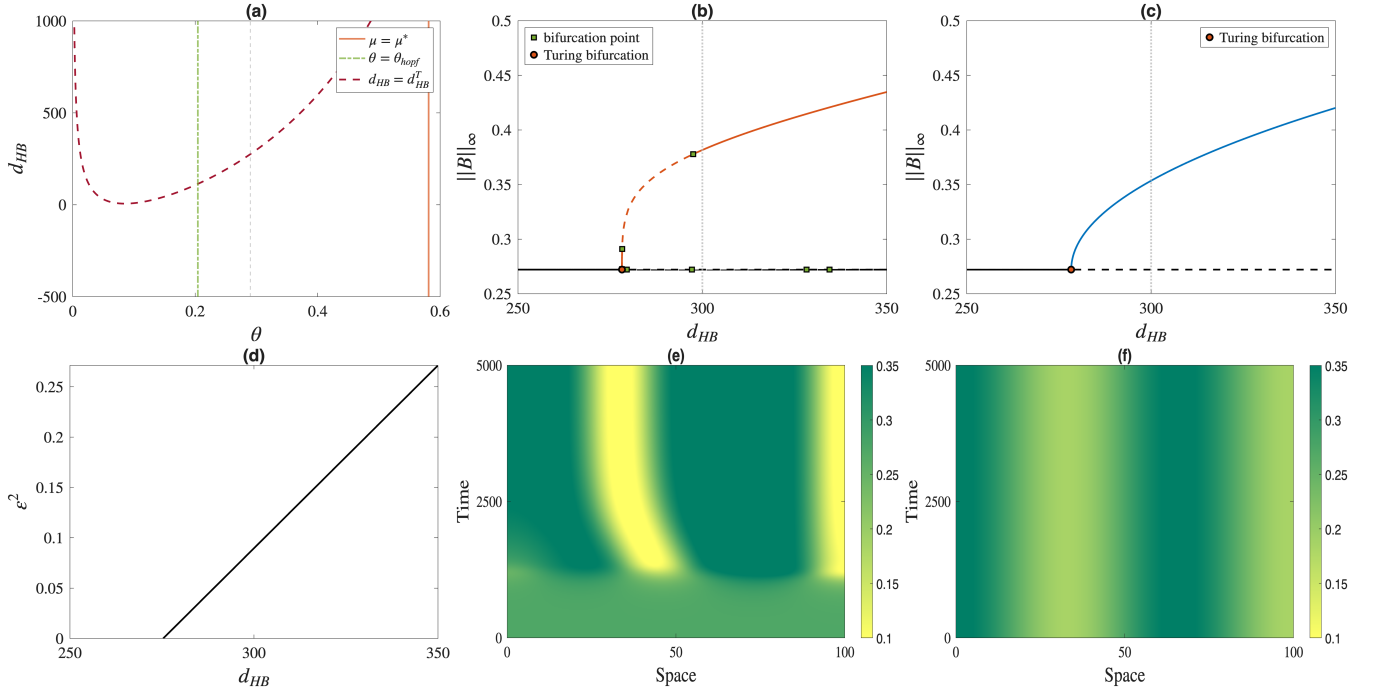}
\caption{(a) Bifurcation diagram in the $(\theta,d_{HB})$-plane for the parameter set of Figure~\ref{fig:weak_regime_diff}(c). 
(b)--(c) Numerical and theoretical bifurcation diagrams. Lines are the same meaning as in Figure~\ref{fig:hopf_amp}. (d) Normalized distance from the bifurcation threshold expressed by $\varepsilon^2$. (e)--(f) Numerical and analytical solutions for the parameter set indicated by the dotted line in panels (b)--(c).}
\label{fig:turing_amp}
\end{figure}

Finally, the most intricate scenario, where temporal oscillations and spatial heterogeneity interact through a Turing-Hopf bifurcation, is addressed. Figure~\ref{fig:turinghopf_amp} summarizes the theoretical and numerical behaviour close to the Turing--Hopf codimension-two point and illustrates how the simultaneous loss of temporal and spatial stability enriches the bifurcation landscape beyond what occurs in the pure Hopf or pure Turing settings. Panel~(a) reports the instability region in the $(\theta,d_{HB})$-plane, highlighting the intersection of the Hopf and Turing loci and thus identifying the parameter configurations for which both critical modes become marginal. Panels~(b)--(c) compare the numerical continuation
results with the weakly nonlinear predictions obtained from the coupled Stuart--Landau system: as in the previous cases, the amplitude equations accurately capture the onset of instability and the initial deformation of the solution branches, while deviations arise once the dynamics move farther from criticality. The normalized distance $\varepsilon^2$ reported in panel~(d) quantifies the range of validity of the asymptotic description, confirming that the approximation remains reliable only in a narrow neighbourhood of the bifurcation point. This behaviour is further reflected in panels~(e)--(l), which display the spatiotemporal evolution of the solution for representative parameter values. In particular, the figure reveals three distinct dynamical regimes, associated with the excitation of one or both critical modes by the chosen initial perturbation. When the disturbance projects mainly onto the temporal eigenspace, the Hopf component dominates and the system develops spatially homogeneous oscillations converging to the single mode branch $(\Omega_{\infty}, 0)$, as shown in panels~(e)--(f). Conversely, if the initial condition excites primarily the spatial eigenmode, the dynamics converge toward stationary spatially periodic patterns characterized by the single mode branch $(0, \Gamma_{\infty})$, displayed in panels~(g)--(h). Finally, when both modes are significantly activated, the full Turing--Hopf interaction emerges, giving rise to mixed spatiotemporal structures characterised by temporal oscillations superimposed on spatial periodicity and converging to the mixed-mode branch $(\Omega_{{TH}_\infty}, \Gamma_{{TH}_\infty})$. These behaviours, illustrated in panels~(i)--(l), reflect the nonlinear coupling between the oscillatory and stationary components predicted by the amplitude equations. Moreover, Figure~\ref{fig:turinghopf_amp} highlights the capacity of the Turing--Hopf mechanism to generate multiple coexisting dynamical responses near criticality, with the specific outcome determined by the structure of the initial perturbation.

\begin{figure}[p!]
\centering
\includegraphics[width=1\textwidth]{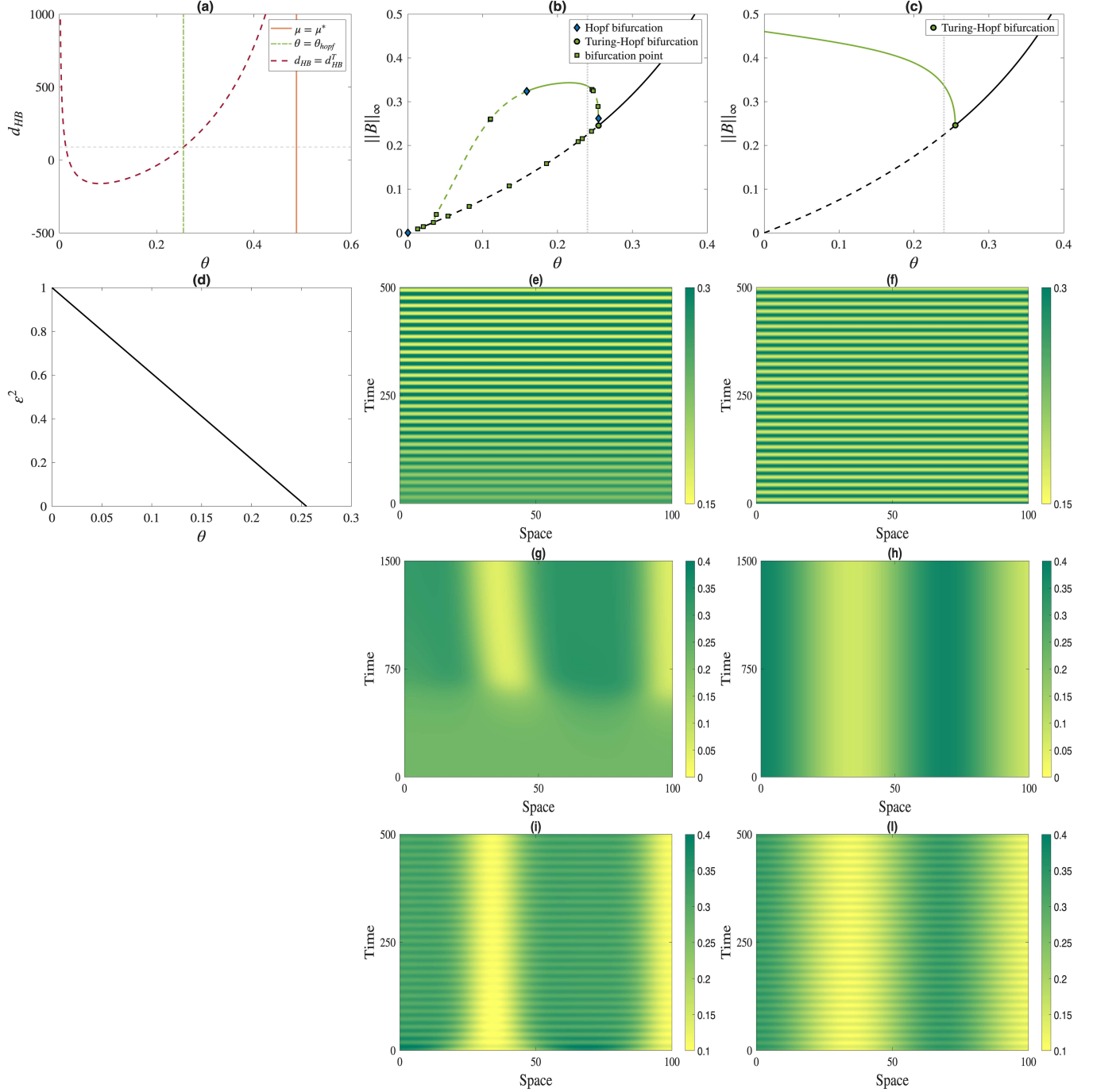}
\caption{(a) Bifurcation diagram in the $(\theta,d_{HB})$-plane for the parameter set of Figure~\ref{fig:strong_regime_diff}(c), highlighting the occurence of a Turing--Hopf codimension-two point. (b)--(c) Numerical and theoretical bifurcation diagrams. Lines have the same meaning as in Figure~\ref{fig:hopf_amp}, with green branches corresponding to Turing-Hopf induced solutions. (d) Normalized distance from the bifurcation threshold expressed by $\varepsilon^2$, quantifying the
range of validity of the weakly nonlinear approximation. (e)--(f) Numerical and analytical solutions, for the parameter set indicated by the dotted line in panels~(b)--(c), illustrating the emergence of spatially homogeneous oscillations (Hopf-dominated). (g)--(h) Numerical and analytical solutions, for the parameter set indicated by the dashed line in panels~(b)--(c), showing stationary spatially periodic patterns (Turing-dominated). (i)--(l) Numerical and analytical solutions, for the parameter set indicated by the dotted line in panels~(b)--(c), for a mixed excitation of the two critical modes, illustrating fully developed Turing--Hopf spatiotemporal patterns.}
\label{fig:turinghopf_amp}
\end{figure}

These findings clearly define the range within which the weakly nonlinear approach remains reliable and emphasize its inherent limitations when the system moves beyond the primary bifurcation, where secondary branches and large-amplitude dynamics prevail.

\newpage 
\section{Conclusion}
\label{sec:disc}

This manuscript aims at understanding how toxin-mediated interactions and movement strategies influence the emergence of coherent structures in plant--herbivore systems. These structures, including stationary spatial patterns, temporal oscillations, and mixed spatiotemporal regimes, are not only mathematically intriguing but also ecologically significant, as they govern the organization and resilience of ecosystems. Therefore, identifying the conditions under which such configurations arise is essential for predicting abrupt transitions and assessing the stability of ecological communities under environmental stress.

To this goal, an extended reaction--diffusion model, that generalizes the classical toxin-dependent functional response by incorporating cross-diffusion effects, is introduced. This formulation captures realistic movement behaviors of herbivores and provides a framework for analyzing how chemical defenses and spatial redistribution interact to destabilize homogeneous states. The analysis firstly focuses on the characterization of homogeneous equilibria and their stability properties. In particular, by distinguishing between weak and strong toxicity regimes, Theorems \ref{th:weak}--\ref{th:strong} establish the conditions for the existence and feasibility of coexistence states, while Corollaries \ref{cr:hopf_weak}--\ref{cr:hopf_strong} identify Hopf bifurcation thresholds leading to oscillatory dynamics. These results help in clarifying how physiological traits, such as mortality and conversion efficiency, modulate the onset of temporal fluctuations.

The role of spatial processes is then examined through the study of diffusion-driven instabilities. The analysis reveals that cross-diffusion can trigger Turing bifurcations, giving rise to stationary vegetation patches and herbivore aggregations. Figures \ref{fig:weak_regime_diff} and \ref{fig:weak_regime_diff_sim} illustrate how small variations in mortality or movement intensity can shift the system from uniform equilibria to patterned landscapes, or even to mixed regimes where oscillations and spatial heterogeneity coexist. This interplay is further explored in Section \ref{sec:close}, where a weakly nonlinear approach is employed to derive amplitude equations near bifurcation thresholds. In this context, the Stuart--Landau framework provides a rigorous description of pattern modulation and highlighted the role of nonlinear saturation and cross-interaction terms in shaping the morphology of emerging structures.

From a broader perspective, it is shown that toxin-mediated feedbacks and movement strategies act as key drivers of complexity in ecological systems. The identification of codimension-two points, where Hopf and Turing instabilities interact, underscores the potential for rich spatiotemporal dynamics, including oscillating stripes and breathing patches. Such configurations reflect ecosystems operating near criticality, where resilience is reduced and small environmental changes may trigger large-scale reorganizations.

An important outcome of this analysis concerns the conditions under which spatial pattern formation can arise in the proposed vegetation-herbivore framework. In particular, it has been shown that nontrivial spatial and spatiotemporal patterns emerge only in the presence of repulsive vegetation effects. While this mechanism plays a crucial role in triggering self-organization, it highlights an intrinsic limitation of the current two-species model: attractive vegetation dynamics do not lead to pattern formation within this setting. From an ecological perspective, this finding suggests that additional mechanisms are required to reproduce richer spatial behaviors. In this direction, the inclusion of further ecological components, such as a third interacting species or additional feedback mechanisms, appears as a natural and promising extension. These generalizations would allow the exploration of more complex scenarios that are closer to the experimental observation.

Future research should address several directions. Extending the analysis to two-dimensional domains will allow the exploration of more intricate spatial morphologies, such as hexagonal lattices or labyrinthine structures, which are commonly observed in semi-arid landscapes \cite{Consolo2025}. A rigorous investigation of coherent structures far from bifurcation thresholds, possibly through geometric singular perturbation theory (GSPT), would provide insight into global dynamics and multi-scale interactions beyond the weakly nonlinear regime \cite{Grifo2025IV}. Finally, the study of secondary bifurcations, including period-doubling cascades and routes to chaos, represents a challenging but promising avenue for understanding complex behaviors and unpredictability in ecological systems.

\section*{Acknowledgments}
This research was funded by MUR (Italian Ministry of University and Research) PNRR - Missione 4, Componente 2, Investimento 1.1 - Bando Prin 2022 PNRR - Decreto Direttoriale No. 1409 del 14-09-2022 through PRIN2022-PNRR Project No. P2022WC2ZZ ``A multidisciplinary approach to evaluate ecosystems resilience under climate change" CUP J53D23015990001 and by INdAM-GNFM. GG is a post-doc fellow supported by the National Institute of Advanced Mathematics (INdAM) founded by the European Union - NEXTGENERATIONEU CUP E63C25000470007.

\section*{Data availability}
No data was used for the research described in the article.

\bibliographystyle{unsrtnat}
\bibliography{bibliography}

\end{document}